\documentclass[runningheads]{llncs}

\usepackage[T1]{fontenc}
\usepackage{graphicx}
\usepackage{amssymb}
\usepackage{placeins}
\usepackage{pifont}
\usepackage{tikz}
\usepackage{xcolor}
\usetikzlibrary{patterns}
\usetikzlibrary{positioning,arrows.meta,fit,calc}
\usepackage{pgfplots}
\pgfplotsset{compat=1.18}
\usepgfplotslibrary{groupplots}
\usepackage{booktabs}
\usepackage{hyperref}
\usepackage{url}
\usepackage{amsmath}
\usepackage{mathtools}
\usepackage[ruled,vlined]{algorithm2e}
\SetCommentSty{small}
\usepackage{multirow}
\usepackage{float}
\usepackage{enumitem}
\setlist[itemize]{noitemsep,nolistsep}

\newcommand{\circledfill}[1]{\tikz[baseline=(C.base)]\node[fill=black,circle,inner sep=0.6pt](C){\scriptsize\color{white}#1};}

\newtheorem{assumption}{Assumption}

\begin{document}

\title{A Bayesian Correlated Equilibrium for Early Insider-Threat Detection}
\titlerunning{Bayesian Correlated Equilibrium for Insider-Threat Detection.}

\author{Javed M. Shah\inst{1}\orcidID{0009-0009-7472-5614} \and Ian A. Kash\inst{1}\orcidID{0000-0002-7826-8555} \and Natalie Parde\inst{1}\orcidID{0000-0003-0072-7499}}
\authorrunning{Shah et al.}
\institute{University of Illinois Chicago, Chicago, United States\\
\email{jshah93@uic.edu, iankash@uic.edu, parde@uic.edu}}

\maketitle

\begin{abstract}
We model insider threat detection as a dynamic Bayesian game in which a platform coordinates a committee of strategic certifiers to sustain equilibrium among honest users and detect malicious deviations before exfiltration. Certifiers and users operate under a Bayesian Temporal Correlated Equilibrium (BTCE), where a sealed-envelope correlating device issues private recommendations over time and obedience is verified at every on-path information state. Unlike Stackelberg formulations, BTCE coordinates heterogeneous certifiers without requiring commitment power. We incorporate present bias and loss aversion to capture impulsive escalation dynamics, enabling 1.7--4.5 days earlier detection than rational baselines. We prove three guarantees: (1)~calibrated intervention losses make recommended behavior a current-self best response despite behavioral biases, (2)~controlled evidence accumulation guarantees intervention in bounded expected time before exfiltration, and (3)~median aggregation confines implemented actions to the honest recommendation range when fewer than half of certifiers are Byzantine. On CERT r6.2 our mechanism achieves up to 28.3\% pre-exfiltration detection with false positives below 1.6\%, while both a transformer baseline and a streaming provenance approximation (HOLMESLite) achieve near-zero pre-exfiltration detection under comparable constraints.
\end{abstract}

\section{Introduction}
\label{sec:introduction}
Insider threats arise when authorized users misuse legitimate access to exfiltrate data, sabotage systems, or abuse privileges. While the detection of these threats remains one of the most intractable challenges in enterprise security, traditional approaches still rely on a theoretical framework that assumes a "rational actor". Existing insider risk modeling uses classical rational paradigms~\cite{LIU200875,joshi21} where malicious insiders are expected-utility maximizers with exponential discounting. Insiders are assumed to trade off immediate payoffs from reconnaissance and exfiltration against the perceived probability and cost of detection. We point to a growing body of empirical and theoretical work that shows systematic deviations from this benchmark. In fact, decision makers are found to often exhibit present bias and time inconsistency, placing disproportionate weight on immediate rewards relative to delayed consequences~\cite{FrederickLoewensteinODonoghue2002}. They also display reference dependence and loss aversion, evaluating outcomes relative to a psychological reference point and becoming risk-seeking in the loss domain~\cite{TverskyKahneman1992}. These phenomena are salient for insider scenarios. Disgruntled employees may perceive themselves to be in a loss domain, and opportunistic insiders may act impulsively when presented with a rare vulnerability, even if they previously intended to behave safely~\cite{Greitzer2018,legg2015visualizing}. 

This gap between modeled and actual insiders often explains why static anomaly detection fails to stop irreversible harm. We approach this as a dynamic multiagent interaction between platform users
and a committee of \emph{certifier} agents. Certifiers can be automated detectors, rule engines, or human reviewers. They  condition on the publicly recorded history \(H_u^t\), maintaining a
local posterior over latent user state, and recommend an intervention \(c_{i,u}^t \in A_c\) for each user \(u\) at time \(t\) (Def.~\ref{def:markov-behavioral}).
The platform implements a single outcome \(d_u^t=g(\mathbf c_u^t)\) by aggregating
\(\mathbf c_u^t\) with an order-statistic rule on \((A_c,\preceq)\). Here \(\preceq\) orders interventions by stringency, from low-friction monitoring to access suspension. Routine regimes use a median-type outcome, while high-risk regimes use a fail-safe max rule (Algorithm~\ref{alg:intervention}).

\paragraph{Motivating High Profile Exfiltration Episode.}
We draw attention to several high profile publicly documented insider episodes that illustrate precisely the out-of-equilibrium deviations that BTCE is designed to intercept. In Tesla--Tripp (2018) \cite{bloomberg2018teslasabotage}, sabotage following a denied promotion can be represented as a behavioral regime shift where an adverse organizational event raised the user's reference point \(r_u^t\) and moved the user into the loss domain. Actions such as code edits and scripted exports he later made are actually coordinated sequences rather than isolated anomalies. Standard security models failed to intercept Tripp because his access privileges remained valid and his behavior, while aggressive, looked like "technician work" to static systems. The BTCE model, however, would ingest the "precipitating event" (reassignment) to update the user's latent Behavioral State ($b_u^t$), dynamically computing an Insider Threat Score (ITS) (Def.~\ref{def:ITS}) and forcing earlier intervention. This would allow certifiers to flag the subsequent action sequence (modifying code, exporting data) as a coordinated attack.
More recently, the attack surface has increased even further with the proliferation of non-human identities such as service accounts, API keys and automated agents. As the boundaries of trust erode, the central challenge is building systems that stay secure when trust itself is weaponized. This paper asks: can we design a system that reliably detects and eliminates malicious actors under extreme sparsity before exfiltration begins?

\paragraph{Contributions}
We introduce a \emph{Bayesian Temporal Correlated Equilibrium} (BTCE) for a telemetry-driven governance mechanism: the BTCE characterizes consistent, history-dependent strategies for certifier recommendations, platform aggregation, and user responses. To keep the equilibrium verification operational, we use an on-path one-shot obedience check: at each realized information state, following the current recommendation must dominate any single-period deviation, with future play returning to the recommended protocol (Lemma~\ref{lem:oneshot}). By explicitly modeling the temporal evolution of
disgruntled users, we make three main contributions.
\circledfill{1} We give interpretable sufficient conditions under which calibrated intervention loss makes the recommended action a current-self best response in a Bayesian Temporal Correlated Equilibrium, even with Prospect-theoretic utilities and quasi-hyperbolic discounting (Theorem~\ref{thm:obedience}). \circledfill{2} We establish early detection guarantees, showing that under positive evidence drift the posterior crosses the intervention threshold in bounded expected time before substantial exfiltration, with behavioral modeling yielding 1.7--4.5 days earlier detection than rational baselines (Theorem~\ref{thm:early-detection}). \circledfill{3} We show that median-based order-statistic aggregation keeps the implemented action within the honest recommendation range whenever fewer than half of certifiers are compromised, preserving both equilibrium enforcement and early detection guarantees (Proposition~\ref{prop:median-robust}). Together, these results yield an auditable\footnote{BTCE is designed as an auditable coordination layer on top of systems enterprises already operate, rather than a replacement for SIEM, EDR, DLP, identity, HR, or human-review workflows.}, drift-robust mechanism that combines heterogeneous certifiers and supports escalation before meaningful data loss even when evidence is sparse.

\section{Related Work}
\label{app:related-work}

Game-theoretic security has largely emphasized Stackelberg~\cite{sinha2018stackelberg} or Bayesian~\cite{BergemannMorris2016} models in which defenders commit and attackers best-respond, while repeated games with private monitoring~\cite{Kandori2002} motivate history-dependent coordination but do not combine causal signal dependencies, behavioral bias, and loyalty-type evolution. BTCE instead uses correlated equilibrium~\cite{Aumann1987} to coordinate heterogeneous certifiers without commitment power, while incorporating present bias~\cite{Laibson1997Hyperbolic,ODonoghueRabin1999}, loss aversion~\cite{kahneman1979prospect,TverskyKahneman1992}, and insider psychology~\cite{Greitzer2018}. ML approaches ranging from Isolation Forest~\cite{gavai2015detecting} and LSTMs~\cite{tuor2017deep} to Transformer-UBS~\cite{elbasheer2024enhancing} often achieve user-level detection but struggle on timing-sensitive early-warning metrics; we use a DBN-inspired evidence accumulation layer~\cite{wall2021bayesian} inside the game-theoretic mechanism, and Byzantine-robust aggregation~\cite{LamportByzantine1982} motivates the median order-statistic rule in Algorithm~\ref{alg:intervention} and Proposition~\ref{prop:median-robust}. Provenance systems such as HOLMES~\cite{milajerdi2019holmes} and KAIROS~\cite{cheng2024kairos} reason over causal host-level telemetry, whereas BTCE targets enterprise control planes (cloud and data access) where coordinated misuse can unfold across users and devices without a single host-level execution chain; we therefore benchmark against a streaming approximation, HOLMESLite, rather than treating provenance reconstruction as the main task. Finally, BTCE operates under sparsity comparable to credit card fraud detection ($\sim 0.17\%$ prevalence~\cite{chawla2002smote,trisanto2020effectiveness}) but avoids class imbalance entirely by combining equilibrium analysis with controlled evidence accumulation rather than reducing features or rebalancing the training distribution.

Section~\ref{sec:btce-derivations} introduces the signal model and game;
Section~\ref{sec:behavioral-prospect} develops the behavioral utilities, type
evolution, and belief updates; Section~\ref{sec:btce} defines the BTCE
equilibrium; Section~\ref{sec:early-detection} analyzes early detection,
graduated intervention, and Byzantine robustness; and
Section~\ref{sec:empirical} presents empirical results.


\section{Bayesian Game Design}
\label{sec:btce-derivations}

\paragraph{Model overview (informal).}
We first give an informal story of the model before the formal specification. Each period, a user $u$ takes an unobserved operational action $x_u^t$ (e.g., normal work, negligent behavior, or malicious steps). A committee of certifiers $i\in\mathcal{C}$ maintain beliefs about the user’s latent \emph{loyalty} state $\theta_u^t$
and latent \emph{behavioral} state $b_u^t$.
The platform observes only telemetry-derived signals $s_u^t$ (\S\ref{sec:signal-generation-dbn}) and the committee’s
\emph{implemented} intervention outcome $d_u^{t-1}$ (e.g., approve, monitor, escalate, suspend).
 Based on the public history $H_u^t$ and their (possibly noisy) local posteriors,
certifiers recommend interventions $c_{i,u}^t\in A_c$.
The platform aggregates recommendations into a single implemented action $d_u^t$ via a rule $g$.
Users then respond to the intervention environment and their own present bias (\S\ref{sec:type-evolution}). Next, we present a definition of the game, observations, strategies and a sample per-period, per-user sequence of play. 

\begin{definition}[Behavioral temporal Bayesian game]
\label{def:markov-behavioral}
A behavioral temporal Bayesian game is a tuple
\[
\mathcal{G}
=
\Bigl(
\mathcal{T},\mathcal{C},\mathcal{U},\Theta_{\mathrm{user}},B,\mu^0,\nu^0,
A_u,A_c,g_\kappa,
\mathcal{S},G_{\mathrm{sig}},P_s,\tau_\Theta,\tau_B,
(U_u^t)_{t\in\mathcal{T}},(U_{c,i}^t)_{t\in\mathcal{T},i\in\mathcal{C}}
\Bigr).
\]
Here $\mathcal{T}=\{0,\ldots,T\}$ is the finite horizon, $\mathcal{C}=\{1,\ldots,m\}$ is the certifier committee, and $\mathcal{U}$ is the user population. The loyalty type space is
\[
\Theta_{\mathrm{user}}=\{\texttt{Loyal},\texttt{Negligent},\texttt{Disgruntled},\texttt{Malicious}\},
\]
with $\theta_u^t\in\Theta_{\mathrm{user}}$ and common prior $\mu^0\in\Delta(\Theta_{\mathrm{user}})$; the finite behavioral state space is $B$, with $b_u^t\in B$ and prior $\nu^0\in\Delta(B)$. Under common priors, each certifier initializes $\mu_{i,u}^0=\mu^0$ and $\nu_{i,u}^0=\nu^0$. Users choose $x_u^t\in A_u$, certifiers recommend $c_{i,u}^t\in A_c$, and $A_c$ is totally ordered by stringency $\preceq$. Writing $\mathbf c_u^t=(c_{i,u}^t)_{i\in\mathcal C}\in(A_c)^m$, the platform implements
$d_u^t=g_{\kappa_u^t}(\mathbf c_u^t)\in A_c,
\kappa_u^t\in\{1,\ldots,m\}$, where $g_\kappa$ is the order-statistic aggregation rule of \S\ref{sec:robust-aggregation}. The public signal space is $\mathcal S=\prod_{j\in\mathcal I}\mathcal S_j=[0,1]^K,\qquad K=|\mathcal I|$,
and $s_u^t=(s_{u,j}^t)_{j\in\mathcal I}\in\mathcal S$ is the period-$t$ telemetry vector. The graph $G_{\mathrm{sig}}=(\mathcal I,E)$ is a directed acyclic graph over signal components, and the signal kernel $P_s(\cdot\mid\theta,b,x,d)$ admits a density that factorizes as $p_s(s_u^t\mid \theta_u^t,b_u^t,x_u^t,d_u^{t-1})
=
\prod_{j\in\mathcal I}
p_{s,j}\!\left(
s_{u,j}^t
\,\middle|\,
s_{u,\mathrm{Parents}(j)}^t,\theta_u^t,b_u^t,x_u^t,d_u^{t-1}
\right)$.

Latent states evolve according to the loyalty type kernel
$\tau_\Theta(\theta_u^{t+1}\mid\theta_u^t,x_u^t,d_u^t)\in\Delta(\Theta_{\mathrm{user}})$, and the behavioral type kernel $\tau_B(b_u^{t+1}\mid b_u^t,\theta_u^t,x_u^t,d_u^t)\in\Delta(B)$. Finally, $(U_u^t)_t$ and $(U_{c,i}^t)_{t,i}$ are user and certifier stage utilities. 
\end{definition}
User utility is the behavioral utility $U_{\mathrm{beh}}$ of Def.~\ref{def:beh-utility}, based on intrinsic payoff $\rho:A_u\to\mathbb R$ and Prospect-theoretic value $v(\cdot)$. Certifier $i$ has per-user utility
$U_c^t(\theta_u^t,b_u^t,c_{i,u}^t,s_u^t,\mu_{i,u}^t)$,
and platform-level utility $U_{c,i}^t =
\sum_{u\in\mathcal U}
U_c^t(\theta_u^t,$ $b_u^t,c_{i,u}^t,s_u^t,\mu_{i,u}^t)$,
where $\mu_{i,u}^t\in\Delta(\Theta_{\mathrm{user}})$ is certifier $i$'s belief over user $u$'s loyalty type. Since the horizon is finite, utilities are undiscounted except for the behavioral present-bias term; equivalently, the baseline discount factor is $\delta=1$, with the extension to $\delta\in(0,1]$ immediate.

Certifiers observe the public history of past signals and implemented committee outcomes
$\{(s_u^\tau,d_u^\tau)\}_{\tau < t}$.
User actions $x_u^t$ are not directly observable; they influence play only through the realized signals. Only the
implemented outcome $d_u^t$ is publicly observed, the individual certifier recommendations $(c_{i,u}^t)_i$ are not. A (possibly mixed) user strategy for user $u$ is a sequence
$\pi_u \coloneqq (\pi_u^t)_{t\in\mathcal{T}}$ where, for each $t$, $\pi_u^t:\ \Theta_{\text{user}}\times B \times \mathcal{H}_u^t \to \Delta(A_u)$,
and $\mathcal{H}_u^t$ denotes the set of feasible public histories $H_u^t$. Because certifiers observe the period-$t$ public signal before acting, a (possibly mixed) certifier strategy for certifier $i$
is a sequence $\sigma_i \coloneqq (\sigma_i^t)_{t\in\mathcal{T}}$ where, for each $t$, $\sigma_i^t:\ \mathcal{H}_u^t \times \mathcal{S} \to \Delta(A_c)$,
so that $\sigma_i^t(H_u^t,s_u^t)$ is the distribution over certifier actions after observing $(H_u^t,s_u^t)$. In the pure-strategy case, $\sigma_i^t(H_u^t,s_u^t)$ is a point mass on an element of $A_c$. $\pi \coloneqq (\pi_u)_{u\in\mathcal{U}}$ and $\sigma \coloneqq (\sigma_i)_{i\in\mathcal{C}}$ denote the strategy profiles. Gameplay (per period $t$, per user $u\in\mathcal{U}$) proceeds as follows: the user moves first, the signal realizes, certifiers recommend after observing $s_u^t$, and transitions occur under $d_u^t$. Only $d_u^t$ enters $P_s$, $\tau_\Theta$, and $\tau_B$; the recommendation profile $\mathbf{c}_u^t$ matters only through the aggregation that produces it. Full timing appears in Table~\ref{tab:update-journey}. 
\begin{table}[t]
\caption{Within-period update journey and time indexing.}
\label{tab:update-journey}
\centering
\footnotesize
\setlength{\tabcolsep}{4pt}
\renewcommand{\arraystretch}{1.18}
\begin{tabular}{|c|p{2.3cm}|p{7.3cm}|p{1.2cm}|}
\hline
\textbf{\#} & \textbf{Actor} & \textbf{Update (inputs $\to$ output)} & \textbf{Visible} \\
\hline

0 &
Public history \newline (All parties) &
Start of period $t$ history: $H_u^t=\bigl((s_u^0,d_u^0),\ldots,(s_u^{t-1},d_u^{t-1})\bigr)$; default $d_u^{-1}=\textsc{NoAct}$ &
$H_u^t$ \\
\hline

1 &
User move \newline (User $u$) &
Given latent state $(\theta_u^t,b_u^t)$ and history $H_u^t$, user samples $x_u^t \sim \pi_u^t(\cdot \mid \theta_u^t,b_u^t,H_u^t)$ &
no \\
\hline

2 &
Signal emission \newline (Telemetry) &
$s_u^t \sim P_s(\cdot\mid\theta_u^t,b_u^t,x_u^t,d_u^{t-1})$; DBN factorization:
$p_s(s_u^t\mid\cdot)=\prod_{j\in\mathcal{I}} p_{s,j}(s_{u,j}^t\mid s_{u,\mathrm{Pa}(j)}^t,\theta,b,x,d_u^{t-1})$ &
$s_u^t$ \\
\hline

3 &
Belief filtering \newline (Common posterior) &
Compute action-marginalized likelihood
$L_u^t(\theta,b)=\sum_{x\in A_u}\pi_u^t(x\mid \theta,b,H_u^t)\,p_s(s_u^t\mid \theta,b,x,d_u^{t-1})$, then update belief
$\widetilde{\beta}_u^t(\theta,b)\propto L_u^t(\theta,b)\,\beta_u^t(\theta,b)$. &
\ $\mu_u^{t,\text{sys}}$ \\
\hline

4 &
ITS gate \newline (Platform) &
Compute scalar risk score,  $\mathrm{ITS}(u,t)=\textstyle\sum_{j=1}^k w_j^{\mathrm{ITS}}\,s_{u,j}^t$ &
\ $\mathrm{ITS}$ \\
\hline

5 &
Certifier recs \newline (Certifiers $i\!\in\!\mathcal{C}$) &
Each certifier outputs recommendation $c_{i,u}^t\in A_c$ &
\ $c_{i,u}^t$ \\
\hline

6 &
Aggregation \newline (Committee) &
Aggregate and implement $d_u^t\coloneqq g(\mathbf{c}_u^t)\in A_c$ &
$\boldsymbol{d_u^t}$ \\
\hline

7 &
Belief propagation \newline (Filter) &
Propagate filtered belief $\widetilde{\beta}_u^t$ through $(\tau_\Theta,\tau_B)$ to obtain $\beta_u^{t+1}$, the prior for period $t\!+\!1$ &
\ $\beta_u^{t+1}$ \\
\hline

8 &
Type evolution \newline (User dynamics) &
Transition to next latent state: $\theta_u^{t+1}\!\sim\!\tau_\Theta(\cdot\mid\theta_u^t,x_u^t,d_u^t)$;\enspace
$b_u^{t+1}\!\sim\!\tau_B(\cdot\mid b_u^t,\theta_u^t,x_u^t,d_u^t)$ &
no \\
\hline
\end{tabular}
\end{table}
We model the telemetry process with a two-slice DBN: 
\begin{definition}[Exploit-chain dependency graph] \label{def:exploit-chain}
\label{sec:signal-generation-dbn}
We model public telemetry $s_u^t\in[0,1]^K$ with a controlled DBN, conditional on $(\theta_u^t,b_u^t$, $x_u^t,d_u^t)$, over the channel ontology used in the CERT r6.2~\cite{CERT2016} dataset, $\mathcal I=\{\textsc{Logon},\textsc{LDAP},\textsc{Device}$, $\textsc{File}, \textsc{Email},\textsc{http}\}$.
Conditional on $(\theta_u^t,b_u^t,x_u^t,d_u^{t-1})$, channels factorize along a stylized kill-chain DAG
$G_{\mathrm{sig}}=(\mathcal I,E_{\mathrm{sig}})$ with \[
E_{\mathrm{sig}}\supseteq
\{(\textsc{Logon},\textsc{File}),(\textsc{Logon},\textsc{Device}),
(\textsc{File},\textsc{Email}), (\textsc{File},\textsc{http})\}.\]
\end{definition}
This yields standard hidden-state filtering updates for the joint posterior $\beta_u^t(\theta,b)$ (Proposition~\ref{prop:belief-update}),
which feeds into (i) certifier recommendation policies and (ii) the platform’s temporal correlating device
that generates recommendations over time (Def.~\ref{def:temporal-device}).
The equilibrium concept (BTCE, Def.~\ref{def:btce}) enforces temporal obedience: given their information,
users and certifiers do not profitably deviate from recommendations. Thus file/device anomalies are interpreted in logon context, and exfiltration channels are interpreted in the context of prior file aggregation. Each channel uses a Beta likelihood whose shape parameters depend on type, behavioral state, and control, allowing the filter to assign higher likelihood to coordinated access--collection--exfil sequences than to isolated negligent anomalies.
\section{Behavioral Utility, Type Evolution and Belief System}
\label{sec:behavioral-prospect}

This section develops the behavioral machinery the rest of the paper depends on.
The introduction argued that real insiders depart from the rational-actor
benchmark through present bias and loss aversion; we now make those forces
precise and tie them to our guarantees. We specify (i)~certifier payoffs that
calibrated sanctions will later target (Def.~\ref{def:certifier-payoff}),
(ii)~a prospect-theoretic, quasi-hyperbolic user utility encoding loss-domain
risk-seeking and time inconsistency (Defs.~\ref{def:beh-utility}--\ref{def:qhd}),
and (iii)~the type- and behavioral-state kernels through which adverse events
move a user toward the loss domain and impulsive escalation
(\S\ref{sec:type-evolution}--\S\ref{sec:beliefs}). Together these supply the current-self objective whose obedience Theorem~\ref{thm:obedience} certifies and
the evidence process whose drift Theorem~\ref{thm:early-detection} bounds. Throughout we
use the within-period timing of Def.~\ref{def:markov-behavioral}: the user action $x^t_u$
precedes the signal $s^t_u \sim P_s(\cdot\mid \theta^t_u,b^t_u,x^t_u,d^{t-1}_u)$,
certifiers act after observing $s^t_u$, and payoffs are evaluated using the
implemented outcome $d^t_u$. 

Fix a baseline intervention $\bar a\in A_c$ and let
$\mathbb{I}_{\mathrm{int}}(d)\coloneqq\mathbb{I}\{\bar a\preceq d\}$
(triggering when $d$ is restrictive enough); write
$\Theta_{\mathrm{benign}}\coloneqq\Theta_{\mathrm{user}}\setminus\{\texttt{Malicious}\}$.
Define detection events using $d_u^t$:
$\mathbb{I}_{\mathrm{det}}(u,t)\coloneqq \mathbb{I}\{\theta_u^t=\texttt{Malicious}\}\,\mathbb{I}_{\mathrm{int}}(d_u^t)$,
$\mathbb{I}_{\mathrm{miss}}(u,t)\coloneqq \mathbb{I}\{\theta_u^t=\texttt{Malicious}\}\,(1-\mathbb{I}_{\mathrm{int}}(d_u^t))$, $\mathbb{I}_{\mathrm{fp}}(u,t)\coloneqq \mathbb{I}\{\theta_u^t\in\Theta_{\mathrm{benign}}\}\,\mathbb{I}_{\mathrm{int}}(d_u^t)$. Let $R_{\mathrm{det}}>0$ be the reward for correctly intervening on a malicious user,
$C_{\mathrm{miss}}>0$ the cost of a miss, and $C_{\mathrm{fp}}>0$ the cost of a false positive. 

\begin{definition}[Certifier stage payoff]
\label{def:certifier-payoff}
The certifier's per-user stage payoff is $\bar U_{c}^{t}\!\bigl(\theta_u^t,b_u^t,c_{i,u}^t,s_u^t,d_u^t\bigr)
\coloneqq
R_{\mathrm{det}}\mathbb{I}_{\mathrm{det}}(u,t)
-
C_{\mathrm{miss}}\mathbb{I}_{\mathrm{miss}}(u,t)
-
C_{\mathrm{fp}}\mathbb{I}_{\mathrm{fp}}(u,t)$. The platform-level certifier stage utility is the sum across users,
$U_{c,i}^{t}\;\coloneqq\;\sum_{u\in\mathcal{U}}$ $\bar U_{c}^{t}\!\bigl(\theta_u^t,b_u^t,c_{i,u}^t,s_u^t,d_u^t\bigr)$.
\end{definition}

\begin{definition}[Behavioral user stage utility]
\label{def:beh-utility}
For $\alpha\in(0,1]$, $\lambda>1$, and reference point $r$, define
$v(z)=z^\alpha$ for $z\ge0$ and $v(z)=-\lambda(-z)^\alpha$ for $z<0$.
Let $\rho:A_u\to\mathbb R$ be intrinsic action payoff, $C>0$ the penalty from detection, and let
$b=(r(b),\psi_{\mathrm{pb}}(b))$ with $r(b)\in\mathbb R_+$ and $\psi_{\mathrm{pb}}(b)\in(0,1]$.
Given $p_{\mathrm{det},u}^t
\coloneqq
p_{\mathrm{det}}(x_u^t,d_u^t,s_u^t,\mu_u^{t,\mathrm{sys}})\in[0,1]$,
the user stage utility is \[U_{\mathrm{beh}}
=
(1-p_{\mathrm{det},u}^t)v(\rho(x_u^t)-r(b_u^t))
+
p_{\mathrm{det},u}^t v(-C-r(b_u^t)).\] We set $U_u^t\equiv U_{\mathrm{beh}}$ in Def.~\ref{def:markov-behavioral}.
\end{definition}

\begin{definition}[Quasi-hyperbolic present-biased evaluation]
\label{def:qhd}
At time $t$, a user with $\psi_{\mathrm{pb}}\in(0,1]$ evaluates a continuation stream
$\{U^\tau\}_{\tau=t}^T$ as
$U^t+\psi_{\mathrm{pb}}\sum_{\tau=t+1}^T U^\tau$. When $\psi_{\mathrm{pb}}=1$, preferences are time-consistent and reduce, under $\delta=1$, to the undiscounted continuation sum.
\end{definition}

Thus $U_{\mathrm{beh}}$ captures reference dependence and loss-domain risk taking, while the quasi-hyperbolic criterion makes the time-$t$ self overweight the current period relative to future sanctions. The continuation and current-self value functionals used for temporal obedience are introduced in \S\ref{sec:temporal-obedience}.

\subsection{Type Evolution}
\label{sec:type-evolution}
Loyalty types evolve through a multinomial-logistic kernel $\tau_\Theta$ in which toxic access drives \texttt{Disgruntled}$\to$\texttt{Malicious}, probing pushes \texttt{Loyal}$\to$\texttt{Disgruntled}, and honest behavior stabilizes the current type, while reverse transitions to \texttt{Loyal} are rare. The behavioral state $(r_u^t,\psi_{\mathrm{pb},u}^t)$ evolves through $\tau_B$: the reference point rises under sustained adverse interventions (deepening the loss domain), and present bias sharpens under stress. Discretization is presented in Appendix~\ref{app:type-evolution}. While loyalty types evolve in response to strategic behavior and organizational feedback,  behavioral states capture psychological dynamics driven by organizational experience. We encode loss-domain effects, stress, and time inconsistency
by letting the behavioral state $b_u^t\in B$ contain at least a reference point and a present-bias parameter. Let $B \subseteq \mathbb{R}\times(0,1]$ and write
$b_u^t \equiv (r_u^t,\ \psi_{\mathrm{pb},u}^t)$,
where $r_u^t$ is the reference point and $\psi_{\mathrm{pb},u}^t$ is the quasi-hyperbolic present-bias parameter
(Def.~\ref{def:beh-utility}).

\subsection{Bayesian Belief System}
\label{sec:beliefs}
At each period $t$, public monitoring reveals the realized signal vector $s_u^t$ and the implemented outcome
$d_u^t$, so all certifiers condition on the same public history $H_u^t$.
Given the controlled signal model $P_s(\cdot\mid \theta,b,x,d)$ and the state-transition kernels, the platform
runs a controlled Bayesian filter to form a canonical posterior over the latent user state
$(\theta_u^t,b_u^t)$.
We write this joint belief as $\beta_u^t\in\Delta(\Theta_{\text{user}}\times B)$, and we use its marginal on
$\Theta_{\text{user}}$ as the system-level loyalty belief
$\mu_{u}^{t,\mathrm{sys}}\in\Delta(\Theta_{\text{user}})$ that drives belief-threshold policies and the BTCE
obedience constraints.

\paragraph{Controlled DBN filter.}
\label{par:emission-likelihood}
Let $d_u^{-1}=\textsc{NoAct}$ and let
$\pi_u^t(\cdot\mid\theta,b,H_u^t)\in\Delta(A_u)$ be the user's possibly mixed policy. Since $x_u^t$ is unobserved, define the action-marginalized likelihood
$L_u^t(\theta,b)
\coloneqq
\sum_{x\in A_u}
\pi_u^t(x\mid\theta,b,H_u^t)\,
p_s(s_u^t\mid\theta,b,x,d_u^{t-1})$,
with $p_s$ factorized as in Def.~\ref{def:markov-behavioral}. Observing $s_u^t$ also updates the posterior over the latent action:\[\omega_u^t(x\mid\theta,b)
\coloneqq
\Pr(x_u^t=x\mid\theta,b,H_u^{t+1})
= 
\frac{
p_s(s_u^t\mid\theta,b,x,d_u^{t-1})\,
\pi_u^t(x\mid\theta,b,H_u^t)
}{
L_u^t(\theta,b)
},\] for $L_u^t(\theta,b)>0$; if the signal is uninformative about $x$, then $\omega_u^t=\pi_u^t$.

\begin{proposition}[DBN-filtered Bayesian belief update]
\label{prop:belief-update}
Fix $u\in\mathcal U$ and $t\in\mathcal T$.
The next-step belief $\beta_u^{t+1}\in\Delta(\Theta_{\text{user}}\times B)$ satisfies
$\beta_u^{t+1}(\theta',b')
=$ 
$\sum_{\theta\in\Theta_{\text{user}}}\sum_{b\in B}$ 
$\widetilde{\beta}_u^t(\theta,b)\; 
\sum_{x\in A_u}\pi_u^t(x\mid \theta,b,H_u^t)\;$  
$\tau_{\Theta}(\theta'\mid \theta,x,d_u^t)\;
\tau_{B}(b'\mid b,\theta,x,d_u^t)$,
 where $\widetilde{\beta}_u^t$ is the filtered belief after observing $s_u^t$, given by:
\[
\widetilde{\beta}_u^t(\theta,b)
=
\frac{L_u^t(\theta,b)\,\beta_u^t(\theta,b)}
{\sum_{\theta''\in\Theta_{\text{user}}}\sum_{b''\in B}
L_u^t(\theta'',b'')\,\beta_u^t(\theta'',b'')},
\]
where $L_u^t(\theta,b)$ is the action-marginalized emission likelihood (\S\ref{par:emission-likelihood}). 
The per-user per-period update costs
$O(|\Theta_{\mathrm{user}}||B||A_u|K)$, where $K=|\mathcal I|=6$. With
$|\Theta_{\mathrm{user}}|=4$ and $|B|\approx10^2$, this scales to
$N\approx10^4$ users on commodity hardware, and the DBN-likelihood computation parallelizes across users.
\end{proposition}

\section{Bayesian--Temporal Correlated Equilibrium}
\label{sec:btce}
With beliefs and behavioral utilities fixed, we now define the solution concept
that coordinates certifiers and users over time. We develop an equilibrium notion
that coordinates heterogeneous certifiers without assuming any of them can
commit (unlike Stackelberg models) and that can be verified locally, at each
on-path information state, rather than over entire strategies. Bayesian--Temporal
Correlated Equilibrium (BTCE, Def.~\ref{def:btce}) provides a disciplined coordination model for this setting, where a sealed-envelope device issues private recommendations that depend on the realized history. Obedience is checked at each on-path information state, and the current recommendation must dominate any one-period deviation followed by a return to the protocol. The two results of this section connect this local check to the rest of the mechanism: Lemma~\ref{lem:oneshot} records the immediate no-one-shot-deviation implication of BTCE obedience, and Theorem~\ref{thm:obedience} gives conditions under which calibrated sanctions make the recommended action a present-biased user's current-self best response. Following Aumann~\cite{Aumann1987} and Bergemann--Morris~\cite{BergemannMorris2016}, we coordinate participants through a \emph{temporal correlating device} issuing private recommendations over time, via the extensive-form ``sealed-envelope'' construction~\cite{von2008extensive}. The correlating device
samples an entire contingent recommendation plan ex ante and reveals only the currently relevant recommendation
when an agent reaches the corresponding information set.
In our setting, the period-$t$ public signal $s_u^t$ is realized after the user acts and before certifiers act, so
recommendations are revealed sequentially within each period (user first, then certifiers conditional on $s_u^t$),
while remaining jointly coupled through the single ex-ante draw.

\begin{definition}[Temporal correlating device (sealed-envelope form)]
\label{def:temporal-device}
Fix a user $u$ and horizon $T$.
A \emph{temporal correlating device} is a probability distribution $\zeta_u$ over
recommendation plans
$\zeta_u \;=\;\Bigl(\,(\zeta_{u,x}^t)_{t=0}^T,\ (\zeta_{u,c}^t)_{t=0}^T\,\Bigr)$,
where for each $t$,
$\zeta_{u,x}^t:\Theta_{\mathrm{user}}\times B\times \mathcal H_u^t \to A_u,
\zeta_{u,c}^t:\mathcal H_u^t\times \mathcal S \to (A_c)^m$
are measurable maps.
At time $0$, the device samples a plan $\zeta_u$ and keeps it fixed.
In period $t$, after the public history $H_u^t$ is realized, the device privately recommends to the user
$\widehat x_u^t \;=\; \zeta_{u,x}^t(\theta_u^t,b_u^t,H_u^t)$.
After the public signal $s_u^t$ is realized and observed, the device privately recommends to the certifiers
$\widehat{\mathbf c}_u^t \;=\; \zeta_{u,c}^t(H_u^t,s_u^t)$.
\end{definition}
The recommendations induce an ``extended'' game where players additionally consider them when choosing actions. For a formal treatment, see Definition~2.2 of~\cite{von2008extensive}. Note that conditioning on the sampled plan $\zeta_u$, recommendations are deterministic. Integrating out $\zeta_u$ induces the history-dependent recommendation kernels
$\zeta_{u,x}^t(\cdot\mid \theta,b,H)$ and $\zeta_{u,c}^t(\cdot\mid H,s)$ as the corresponding
pushforward distributions of $\zeta_{u,x}^t(\theta,b,H)$ and $\zeta_{u,c}^t(H,s)$. If we fix a sealed-envelope temporal correlating device $\zeta_u$ and platform rule
$d_u^t=g_{\kappa_u^t}(\mathbf c_u^t)$, $\Phi_u^{\zeta_u}$ denotes the distribution over complete play paths (induced path measure). 
Expectations in Def.~\ref{def:btce} are taken with respect to $\Phi_u^{\zeta_u}$. The equilibrium concept treats the device as committing ex ante to a recommendation plan and revealing only the relevant recommendation as play unfolds, requiring obedience at every on-path history. 

\begin{definition}[Bayesian-temporal correlated equilibrium (BTCE)]
\label{def:btce}
Fix a user $u\in\mathcal U$ and a temporal correlating device $\zeta$ (Def.~\ref{def:temporal-device}).
Let $\Phi_u^{\zeta}$ denote the induced path measure under obedient play. Let beliefs be given by the regular conditional distributions induced by $\Phi_u^{\zeta}$ on the
extended on-path information available to each player (i.e., public history together with the realized
recommendations they have observed), using Bayes' rule wherever the relevant conditioning event has positive probability. For the user, define the recommendation history up to (but not including) $t$ as
$\widehat x_u^{0:t-1} \;\coloneqq\; (\widehat x_u^0,\widehat x_u^1,\ldots,\widehat x_u^{t-1}) \in (A_u)^t$,
and the corresponding augmented user information state as
$\widehat I_{u,t} \;\coloneqq\; (H_u^t,\widehat x_u^{0:t-1})$.
For certifier $i$, define
$\widehat c_{i,u}^{0:t-1} \;\coloneqq\; (\widehat c_{i,u}^0,\widehat c_{i,u}^1,\ldots,\widehat c_{i,u}^{t-1}) \in (A_c)^t,
\widehat I_{i,u,t} \;\coloneqq\; (H_u^t,s_u^t,\widehat c_{i,u}^{0:t-1})$.
For an information state and recommended action, define the continuation payoff under obedience as the
conditional expected total payoff from period $t$ onward (with the player obeying from $t$ onward):
\begin{align*}
\mathfrak{U}_u(t,\theta,b,H,\widehat x^{0:t-1};\hat x)
\;\coloneqq\;
\mathbb E_{\Phi_u^{\zeta}}\!\left[\,U_u \,\middle|\,
\begin{aligned}
&\theta_u^t=\theta,\ b_u^t=b,\ H_u^t=H,\\
&\widehat x_u^{0:t-1}=\widehat x^{0:t-1},\ \widehat x_u^t=\hat x
\end{aligned}
\right].
\end{align*}

Similarly, for certifier $i$,
\begin{align*}
\mathfrak{U}_i(t,H,s,\widehat c_i^{0:t-1};\hat c_i)
\;\coloneqq\;
\mathbb E_{\Phi_u^{\zeta}}\!\left[\,U_i \,\middle|\,
\begin{aligned}
&H_u^t=H,\ s_u^t=s,\\
&\widehat c_{i,u}^{0:t-1}=\widehat c_i^{0:t-1},\ \widehat c_{i,u}^t=\hat c_i
\end{aligned}
\right].
\end{align*}

Similarly, let $\mathfrak{U}_u(t,\theta,b,H,\widehat x^{0:t-1}; x')$ (resp.\
$\mathfrak{U}_i(t,H,s,\widehat c_i^{0:t-1}; c_i')$) denote the same conditional
expectation when the player deviates \emph{only} at time $t$ to $x'\in A_u$
(resp.\ $c_i'\in A_c$) and obeys thereafter.
We say $\zeta$ is a \emph{BTCE} if the following obedience constraints hold at every on-path information state:

(i) User obedience.\label{sec:user-obed}
For every $t\in\{0,\dots,T\}$, every $(\theta,b,H,\widehat x^{0:t-1},\hat x)$ such that
$\Phi_u^{\zeta}\!\Bigl(
\theta_u^t=\theta,\ b_u^t=b,\ H_u^t=H,\ 
\widehat x_u^{0:t-1}=\widehat x^{0:t-1},\ \widehat x_u^t=\hat x
\Bigr)>0$,
and every $x'\in A_u$,
$\mathfrak{U}_u(t,\theta,b,H,\widehat x^{0:t-1};\hat x)
\;\ge\;
\mathfrak{U}_u(t,\theta,b,H,\widehat x^{0:t-1};x')$.

(ii) Certifier obedience.
\label{sec:cert-obed}
For every certifier $i\in\mathcal C$, every $t\in\{0,\dots,T\}$, every $(H,s,\widehat c_i^{0:t-1},\hat c_i)$ such that
$\Phi_u^{\zeta}\!\Bigl(
H_u^t=H,\ s_u^t=s,\ 
\widehat c_{i,u}^{0:t-1}=\widehat c_i^{0:t-1},\ \widehat c_{i,u}^t=\hat c_i
\Bigr)>0$,
and every $c_i'\in A_c$,
$\mathfrak{U}_i(t,H,s,\widehat c_i^{0:t-1};\hat c_i)
\;\ge\;
\mathfrak{U}_i(t,H,s,\widehat c_i^{0:t-1};c_i')$.
\end{definition}

The sealed-envelope construction lets recommendations be correlated across agents and time while preserving the within-period information order. The user acts, the public signal is realized, and then certifiers act conditional on $(H_u^t,s_u^t)$. Under the finite-horizon and bounded-utility assumptions used below, this construction induces a well-defined path measure $\Phi_u^{\zeta}$.

\subsection{Temporal Obedience and Current-Self Best Response}
\label{sec:temporal-obedience}
We formalize a continuation value and a present-biased current-self objective. Throughout, fix a user $u$, a user policy $\pi_u=(\pi_u^t)_{t\in\mathcal{T}}$, and a certifier profile $\sigma$, which we suppress from the value functions. Write the payoff-relevant \emph{information state} $\mathcal{K}_u^t\coloneqq(\theta_u^t,b_u^t,H_u^t)$ (generic value $\mathcal{K}=(\theta,b,H)$), abbreviate the public belief $\mu_u^{t,\mathrm{sys}}\coloneqq\mu^{\mathrm{sys}}(H_u^t)$, and let $H_u^{t+1}=(H_u^t,(s_u^t,d_u^t))$ append the realized pair.

\begin{definition}[Continuation and current-self values]
\label{def:gbbs}
With $V_u^{T+1}\equiv 0$, the continuation value is
\[
V_u^{t}(I)\coloneqq
\mathbb{E}\Bigl[\,U_{\mathrm{beh}}\bigl(x_u^t,d_u^t,s_u^t,\mu_u^{t,\mathrm{sys}},b\bigr)
+V_u^{t+1}\bigl(\mathcal{K}_u^{t+1}\bigr)\ \Big|\ \mathcal{K}_u^t=\mathcal{K}\Bigr],
\]
where the period-$t$ variables follow the on-path kernels: $x_u^t\sim\pi_u^t(\cdot\mid \mathcal{K})$, $s_u^t\sim P_s(\cdot\mid\theta,b,x_u^t,d_u^{t-1})$, $c_{i,u}^t\sim\sigma_i^t(\cdot\mid H,s_u^t)$ with $d_u^t=g_{\kappa_u^t}(\mathbf{c}_u^t)$, and $\theta_u^{t+1}\sim\tau_\Theta(\cdot\mid\theta,x_u^t,d_u^t)$, $b_u^{t+1}\sim\tau_B(\cdot\mid b,\theta,x_u^t,d_u^t)$. The present-biased current-self objective is $V_u^t$ with the current action fixed to $x$ and the continuation weighted by $\psi_{\mathrm{pb}}(b)$:
\[
W_u^{t}(\mathcal{K};x)\coloneqq
\mathbb{E}\Bigl[\,U_{\mathrm{beh}}\bigl(x,d_u^t,s_u^t,\mu_u^{t,\mathrm{sys}},b\bigr)
+\psi_{\mathrm{pb}}(b)\,V_u^{t+1}\bigl(\mathcal{K}_u^{t+1}\bigr)\ \Big|\ \mathcal{K}_u^t=\mathcal{K},\ x_u^t=x\Bigr].
\]
A present-biased best response is any $\pi_u^\star$ such that, with $V$ evaluated under $\pi_u^\star$, $\;\mathrm{supp}\,\pi_u^{\star,t}(\cdot\mid \mathcal{K})\subseteq\arg\max_{x\in A_u}W_u^{t}(\mathcal{K};x)$ for all $(t,\mathcal{K})$.
\end{definition}
Users evaluate continuation streams using $U_{\mathrm{beh}}$ (Def.~\ref{def:beh-utility}) and the present-biased objective $W$ (Def.~\ref{def:gbbs}). Honest certifiers use the per-user stage payoffs in Def.~\ref{def:certifier-payoff} and evaluate continuation streams by the induced finite-horizon sum of stage payoffs. Accordingly, the obedience inequalities in Def.~\ref{def:btce} hold at every on-path information state. This yields the structural implication below.

\begin{lemma}[No profitable on-path one-shot deviation]
\label{lem:oneshot}
If $\zeta$ is a BTCE, then at every $\Phi_u^{\zeta}$-on-path information state no user or certifier has a profitable one-shot deviation from the revealed recommendation.
\end{lemma}

\begin{proof}
Immediate from the user- and certifier-obedience constraints in Definition~\ref{def:btce}, which compare obeying the revealed recommendation with any one-shot deviation $x'\in A_u$ or $c'_i\in A_c$ at each on-path information state.
\end{proof}

\begin{assumption}[Risk--reward alignment and bounded continuation]
\label{ass:alignment}
Fix an on-path user information state $(t,\theta,b,H,\hat x^{0:t-1},\hat x)$ under
a BTCE, with on-path detection risk $\bar p_t(\cdot)$ and intrinsic stage payoff
$\rho(\cdot)$, so that the expected stage utility at sanction level $C\ge 0$ is
$\rho(x)-C\,\bar p_t(x)$ (Def.~\ref{def:gbbs}), and write
$V^{t+1}(x):=\mathbb{E}\!\left[V_u^{t+1}(\mathcal{K}_u^{t+1})\mid \mathcal{K},\,x_u^t=x\right]$
for the expected continuation it induces (this already integrates the detection
outcome $d_u^t$). Assume the aggressiveness order $\preceq$ is total on $A_u$,
so every feasible $x\neq\hat x$ satisfies $x\succ\hat x$ or $x\prec\hat x$.
\begin{description}
\item[(A) Risk--reward alignment.] There exist payoff-gap constants
$\Delta\rho^+,\Delta\rho^-\ge 0$ and risk-change constants
$\Delta q_L^+,\Delta q_U^->0$ such that
\begin{itemize}
  \item every more-aggressive deviation $x^+\succ\hat x$ satisfies
        $\rho(x^+)-\rho(\hat x)\le \Delta\rho^+$ and
        $\bar p_t(x^+)-\bar p_t(\hat x)\ge \Delta q_L^+$;
  \item every less-aggressive deviation $x^-\prec\hat x$ satisfies
        $\rho(\hat x)-\rho(x^-)\ge \Delta\rho^-$ and
        $\bar p_t(\hat x)-\bar p_t(x^-)\le \Delta q_U^-$.
\end{itemize}
\item[(B) Bounded continuation effect.] There exists $\Delta V<\infty$ bounding the
\emph{present-bias-weighted} continuation effect of the current action:
$\psi_{\mathrm{pb}}(b)\,\bigl|V^{t+1}(x)-V^{t+1}(\hat x)\bigr|\le \Delta V$ for every
feasible $x$ at this information state.
\item[(C) Separation.] $
\frac{\Delta\rho^+ + \Delta V}{\Delta q_L^+}\;\le\;
\frac{\Delta\rho^- - \Delta V}{\Delta q_U^-}$.
\end{description}
\end{assumption}

\begin{theorem}[Sufficient condition for current-self obedience]
\label{thm:obedience}
Suppose (A)--(C) hold. Then the interval
$[C_{\min},C_{\max}]
:=\Bigl[\tfrac{\Delta\rho^+ + \Delta V}{\Delta q_L^+},\,
        \tfrac{\Delta\rho^- - \Delta V}{\Delta q_U^-}\Bigr]$
is nonempty, and for every sanction level $C\in[C_{\min},C_{\max}]$ the
recommended action $\hat x$ is a current-self best response (objective $W$,
Def.~\ref{def:gbbs}) at this information state. Equivalently, (C) reads
$\frac{\Delta\rho^-}{\Delta q_U^-}-\frac{\Delta\rho^+}{\Delta q_L^+}
\;\ge\;\Delta V\!\left(\frac{1}{\Delta q_L^+}+\frac{1}{\Delta q_U^-}\right)$,
and in particular requires $\Delta\rho^-\ge\Delta V$ (strictly whenever
$\Delta\rho^++\Delta V>0$): the reward-per-unit-risk forfeited by caution must
exceed that gained by aggression by a margin of at least
$\Delta V\bigl(\tfrac{1}{\Delta q_L^+}+\tfrac{1}{\Delta q_U^-}\bigr)$, covering the
continuation slack in both directions. Without (C), only one-sided obedience is
guaranteed.
\end{theorem}

\begin{remark}
The constants $\Delta\rho^\pm,\Delta q_L^+,\Delta q_U^-,\Delta V$ are evaluated
at a fixed on-path BTCE continuation. Thus $[C_{\min},C_{\max}]$ is a local
obedience interval. Interpreting it as a design range requires the same
continuation bound $\Delta V$ to hold uniformly for all sanction levels in that
range.
\end{remark}

\begin{proof}
Write the current-self deviation gain as
$\Delta W(x)=\bigl[\rho(x)-\rho(\hat x)\bigr]
-C\bigl[\bar p_t(x)-\bar p_t(\hat x)\bigr]
+\psi_{\mathrm{pb}}(b)\bigl[V^{t+1}(x)-V^{t+1}(\hat x)\bigr]$,
whose last term is bounded by $\Delta V$ in absolute value by (B). Since
$\Delta\rho^+,\Delta V\ge0$ and $\Delta q_L^+>0$ we have $C_{\min}\ge0$; hence any
$C\ge C_{\min}$ is nonnegative, and multiplying the risk inequalities of (A) by
$C$ preserves their direction. For a more-aggressive $x^+\succ\hat x$, (A) gives
reward gain $\le\Delta\rho^+$ and risk increase $\ge\Delta q_L^+$, so
$\Delta W(x^+)\le\Delta\rho^+-C\Delta q_L^++\Delta V$; this upper bound is
nonpositive iff $C\ge C_{\min}$, hence $\Delta W(x^+)\le0$ for every such $C$.
For a less-aggressive $x^-\prec\hat x$, (A) gives reward forfeit $\ge\Delta\rho^-$
and risk reduction $\le\Delta q_U^-$, so
$\Delta W(x^-)\le-\Delta\rho^-+C\Delta q_U^-+\Delta V$; this upper bound is
nonpositive iff $C\le C_{\max}$, hence $\Delta W(x^-)\le0$ for every such $C$.
Condition~(C) is exactly $C_{\min}\le C_{\max}$, so $[C_{\min},C_{\max}]$ is
nonempty; fix any $C$ in it. By totality of $\preceq$ on $A_u$, every feasible
$x\neq\hat x$ is some $x^+\succ\hat x$ or some $x^-\prec\hat x$, so the two cases
exhaust all deviations and $\Delta W(x)\le0$ for all $x\in A_u$. Hence
$\hat x\in\arg\max_{x\in A_u}W(\mathcal{K};x)$, i.e.\ a current-self best response.
\end{proof}

When malicious users deviate, evidence accumulates to trigger intervention in finite expected time before exfiltration. We prove this under a positive-drift assumption on the per-period log-likelihood ratio (Theorem~\ref{thm:early-detection}).

\section{Early Detection Mechanism}
\label{sec:early-detection}

We now develop an operational early-warning policy from the BTCE specification. In practice, a platform observing sparse, noisy telemetry must decide when to intervene before exfiltration is complete. Concretely, we (i) map the system belief, or its calibrated score proxy $\mathrm{ITS}(u,t)$ (Def.~\ref{def:ITS}), to an operational decision rule; (ii) specify how a committee of certifiers aggregates recommendations into a single intervention $d_u^t$; and (iii) define the resulting intervention dynamics over time. The platform maintains a joint posterior $\beta_u^t(\theta,b)\in\Delta(\Theta_{\text{user}}\times B)$ updated by the DBN filter of Proposition~\ref{prop:belief-update}. This belief drives both calibrated sanctions for equilibrium enforcement and evidence accumulation for detection.

\begin{definition}[Insider Threat Score (ITS)]
\label{def:ITS}
The Insider Threat Score at time $t$ for user $u$ is a weighted linear combination of signal components: $\mathrm{ITS}(u,t)
=
(\mathbf{w}^{\mathrm{ITS}})^\top s_u^t
=
\sum_{j=1}^k w_j^{\mathrm{ITS}}\, s_{u,j}^t$, where $\mathbf{w}^{\mathrm{ITS}}=(w_1^{\mathrm{ITS}},\ldots,w_k^{\mathrm{ITS}})\in\mathbb{R}^k$ is a weight vector calibrated from historical telemetry (e.g., breach labels) to maximize discrimination between malicious and benign behavior patterns. Formally, we do not assume an exact identity between $\mu_{u}^{t,\mathrm{sys}}$ and $\mathrm{ITS}$; rather,
calibration is used operationally to select thresholds that correspond to desired posterior risk levels. 
\end{definition}
Note that the calibration map $f$ can be estimated via logistic regression on labeled telemetry.
Let $y_u^t\in\{0,1\}$ indicate whether user $u$ is malicious at time $t$.
Using features $s_u^t$, a fitted model takes the form $\log\frac{\Pr(y_u^t=1\mid s_u^t)}{\Pr(y_u^t=0\mid s_u^t)}
=
\alpha_0 + \alpha_1\,(\mathbf{w}^{\mathrm{ITS}})^\top s_u^t
=
\alpha_0 + \alpha_1\,\mathrm{ITS}(u,t)$.
Thus we may set
$\mu_{u}^{t,\mathrm{sys}}(\texttt{Malicious})
\approx
\operatorname{logistic}\!\bigl(\alpha_0+\alpha_1\,\mathrm{ITS}(u,t)\bigr)$. So thresholding on $\mathrm{ITS}(u,t)$ is approximately equivalent to thresholding on the system belief
$\mu_{u}^{t,\mathrm{sys}}(\texttt{Malicious})$.

\subsection{Tiered intervention and evidence drift}
\label{sec:robust-aggregation}
Fix thresholds $-\infty<\eta_{\mathrm{low}}<\eta_{\mathrm{high}}<\infty$ that partition $\mathrm{ITS}(u,t)$ into low-, mid-, and high-risk tiers. Let $(A_c,\preceq)$ be the totally ordered intervention space, where $a\preceq a'$ means that $a'$ is weakly more restrictive, and take
$\textsf{NoAct}\ \prec\ \textsf{Monitor}\ \prec\ \textsf{Escalate}\ \prec\ \textsf{SuspendAccess}$, with $\textsf{Monitor}$ denoting passive low-friction logging or heightened observation without access restriction. Given committee recommendations $\mathbf c_u^t=(c_{1,u}^t,\ldots,c_{m,u}^t)$ $\in(A_c)^m$, write
$c_u^{t,(1)}\preceq\cdots\preceq c_u^{t,(m)}$ for the order statistics and define
$g_\kappa(\mathbf c_u^t)\coloneqq c_u^{t,(\kappa)}$, $\kappa\in\{1,\ldots,m\}$. Algorithm~\ref{alg:intervention} maps filtered belief and signal evidence to an implemented action: low-risk users receive $\textsf{NoAct}$, isolated mid-risk anomalies receive $\textsf{Monitor}$, coordinated mid-risk cases use median aggregation $\kappa=\lceil m/2\rceil$, and high-risk cases use the fail-safe max rule $\kappa=m$. The DBN posterior gives the dominant type estimate $\hat\theta_u^t=\arg\max_{\theta\in\Theta_{\mathrm{user}}}\mu_u^{t,\mathrm{sys}}(\theta)$, while the kill-chain flag $\rho_u^t$ records whether the current signal pattern contains coordinated exploit-chain evidence.

\begin{algorithm}[h]
\caption{Temporal Intervention Protocol}
\label{alg:intervention}
\DontPrintSemicolon
\SetKwInOut{Input}{Input}
\SetKwInOut{Output}{Output}

\Input{$\mathrm{ITS}(u,t)$, belief $\mu_u^{t,\mathrm{sys}}\in\Delta(\Theta_{\mathrm{user}})$, kill-chain flag $\rho_u^t\in\{0,1\}$\; certifier committee size $m$, recommendations $\mathbf{c}_u^t\in(A_c)^m$}
\Output{Implemented outcome $d_u^t\in A_c$, aggregation index\; $\kappa_u^t$, type estimate $\hat\theta_u^t$}
$\kappa_u^t \leftarrow \lceil m/2\rceil$, $\hat\theta_u^t \leftarrow \arg\max_{\theta\in\Theta_{\mathrm{user}}}\,\mu_u^{t,\mathrm{sys}}(\theta)$\tcp{\scriptsize dominant type from DBN posterior}

\tcp{--- Type-conditioned intervention ---}

\uIf{$\mathrm{ITS}(u,t) < \eta_{\mathrm{low}}$}{
  $d_u^t \leftarrow \textsf{NoAct}$\;
  \tcp{\textsc{Loyal}: routine ops, all channels stable, $\theta_u^t$ held at \texttt{Loyal}}
  \tcp{\textsc{Negligent}: isolated minor anomaly; no kill-chain follow-on observed}
}
\uElseIf{$\mathrm{ITS}(u,t) < \eta_{\mathrm{high}}$
         \textbf{and} $\rho_u^t = 0$
         \textbf{and} $\hat\theta_u^t \neq \texttt{Malicious}$}{
  $d_u^t \leftarrow \textsf{Monitor}$\;
  \tcp{\textsc{Negligent}: uncoordinated spike, no cross-channel progression}
}
\uElseIf{$\mathrm{ITS}(u,t) < \eta_{\mathrm{high}}$}{
  $d_u^t \leftarrow g_{\kappa_u^t}(\mathbf{c}_u^t)$\tcp*[f]{median aggregation in mid-risk tier}
  
  \tcp{\textsc{Disgruntled}: correlated x-channel escalation, $r_u^t\!\uparrow$, loss-domain entry;}
  \tcp{\textsc{Malicious} pre-threshold ($\rho_u^t{=}1$,
       $\mathrm{ITS} < \eta_{\mathrm{high}}$): evidence building}
}
\Else{
  $\kappa_u^t \leftarrow m$\;
  $d_u^t \leftarrow g_{\kappa_u^t}(\mathbf{c}_u^t)$ \tcp*[f]{\textsc{Malicious}: high-risk intervention       (\textsc{Logon}$\!\to\!$\textsc{File}$\!\to\!$\textsc{Exfil})}
}
\Return $(d_u^t,\,\kappa_u^t,\,\hat\theta_u^t)$\;
\end{algorithm}

\paragraph{Evidence drift and escalation time.} 
\label{sec:early-detection-llr} 
Here we formalize the operational intuition that interventions shape future telemetry, telemetry updates the malicious posterior, and escalation occurs once the posterior crosses a chosen threshold. Even when the available evidence is sparse, as long as the signals generated under truly malicious behavior contribute evidence at a reliably positive rate, the expected time until the system’s belief crosses \(\mu_{\mathrm{thr}}\) is bounded. This implies escalation will occur in finite expected time and potentially before data loss. 
We use posterior log-odds to study escalation timing ($M\coloneqq \texttt{Malicious}$): $\Lambda_u^t\coloneqq\operatorname{logit}\bigl(\mu_u^{t,\mathrm{sys}}(M)\bigr)$. The predictive log-likelihood ratio increment $\ell_u^t$ satisfies the Bayes recursion $\Lambda_u^{t+1}=\Lambda_u^t+\ell_u^t$, where the previously implemented control $d_u^{t-1}$ shapes the period-$t$ signal model and hence the malicious drift $\mathbb E[\ell_u^t\mid H_u^t,\theta_u^t=M]$. Fix $\mu_{\mathrm{thr}}\in(0,1)$, define $t_{\mathrm{intervene}}\coloneqq \inf\{t\ge t_0:\mu_u^{t,\mathrm{sys}}(M)\ge\mu_{\mathrm{thr}}\}$, \[\Delta t^\star \coloneqq \frac{\operatorname{logit}(\mu_{\mathrm{thr}})-\operatorname{logit}(\mu_u^{t_0,\mathrm{sys}}(M))} {\lambda_{\min}}.\]

\begin{theorem}[Early detection under controlled evidence accumulation]
\label{thm:early-detection}
Suppose $\Lambda_u^{t_0}<\operatorname{logit}(\mu_{\mathrm{thr}})$ and there exist
$\lambda_{\min}>0$ and $\ell_{\max}<\infty$ such that, before intervention,
$\mathbb{E}[\ell_u^t\mid H_u^t,\theta_u^t=M]\ge\lambda_{\min}$ and
$|\ell_u^t|\le \ell_{\max}$ for all $t\ge t_0$. Then
$\mathbb{E}[t_{\mathrm{intervene}}-t_0\mid H_u^{t_0},\theta_u^{t_0}=M]
\le
\Delta t^\star+\frac{\ell_{\max}}{\lambda_{\min}}$.
\end{theorem}

\begin{proof}[sketch]
Let $a=\operatorname{logit}(\mu_{\mathrm{thr}})$ and
$\gamma=t_{\mathrm{intervene}}$. For $\tau_n=\gamma\wedge n$, summing the stopped
drift inequality gives
$\lambda_{\min}\mathbb{E}[\tau_n-t_0]
\le
\mathbb{E}[\Lambda_u^{\tau_n}]-\Lambda_u^{t_0}$.
Before crossing, $\Lambda_u^t<a$; at the crossing, bounded increments imply
$\Lambda_u^\gamma\le a+\ell_{\max}$. Hence
$\mathbb{E}[\Lambda_u^{\tau_n}]\le a+\ell_{\max}$, so
$\mathbb{E}[\tau_n-t_0]
\le
\frac{a-\Lambda_u^{t_0}+\ell_{\max}}{\lambda_{\min}}
=
\Delta t^\star+\frac{\ell_{\max}}{\lambda_{\min}}$.
Letting $n\to\infty$ and using monotone convergence gives the claim.
\end{proof}

\begin{remark}[From obedience to exclusion: the post-intervention continuation game]
\label{rem:post-intervention-game}
If a malicious user deviates and the evidence crosses the intervention region, the platform implements \textsc{SuspendAccess}: user $u$ is inactive from period $t{+}1$ onward with continuation value normalized to zero, and the continuation game is the BTCE environment restricted to the active set $\mathcal U_{\mathrm{act}}^{t+1}\coloneqq\{v\in\mathcal U:d_v^t\neq\textsc{SuspendAccess}\}$, a reduced game on the remaining population.
\end{remark}

\paragraph{Interpretation and design takeaway.}
The theorem converts positive malicious evidence drift into a timeline guarantee: if the controlled LLR has average increment at least $\lambda_{\min}>0$, then the posterior reaches the intervention threshold in bounded expected time. In BTCE, correlated access--collection--exfil signals make malicious trajectories systematically more likely under the attack model than the benign model, so behavioral escalation appears empirically as earlier threshold crossing relative to the rational ablation.

\paragraph{Byzantine-Resilient Aggregation}

Let $\mathcal C=\{1,\dots,m\}$ index certifiers and let $\mathcal C_{\mathrm{byz}}\subseteq\mathcal C$
be an unknown set of Byzantine certifiers; define $\mathcal C_{\mathrm{hon}}=\mathcal C\setminus\mathcal C_{\mathrm{byz}}$.
Assume fewer than half are Byzantine: $|\mathcal C_{\mathrm{byz}}|<m/2$.

\begin{proposition}[Robustness of the median; fail-safe of the max]
\label{prop:median-robust}
Fix any profile of recommendations $\mathbf c_u^t\in (A_c)^m$ and write its order statistics
$c_u^{t,(1)}\preceq\cdots\preceq c_u^{t,(m)}$. If $\kappa=\lceil m/2\rceil$, then the implemented action $d_u^t=g_{\kappa}(\mathbf c_u^t)$ lies within the range of honest recommendations: $\min\nolimits_{\preceq}\{c_{i,u}^t:i\in\mathcal C_{\mathrm{hon}}\}
\ \preceq\ 
d_u^t
\ \preceq\
\max\nolimits_{\preceq}\{c_{i,u}^t:i\in\mathcal C_{\mathrm{hon}}\}$. Further, if $\kappa=m$, then $d_u^t=g_m(\mathbf c_u^t)=\max_{\preceq}\{c_{i,u}^t:i\in\mathcal C\}$.
\end{proposition}

\paragraph{Anchor-clip ablation.}
The median rule is the headline mechanism certified by Proposition~\ref{prop:median-robust}; we also report a score-level anchor-clip ablation to separate aggregation-family effects from the certified median rule. Let $\widetilde\mu_u^t=\sum_b\widetilde\beta_u^t(\texttt{Malicious},b)$ be the filtered malicious marginal and let $\gamma=0.05$; for scalar reports $\mathbf c_u^t\in[0,1]^m$, define
$g^{\mathrm{ac}}(\mathbf c_u^t;\widetilde\mu_u^t,\gamma)
=
\frac{1}{m}\sum_{i=1}^m
\Pi_{[\widetilde\mu_u^t-\gamma,\widetilde\mu_u^t+\gamma]}(c_{i,u}^t)$.
Since every clipped report lies in $[\widetilde\mu_u^t-\gamma,\widetilde\mu_u^t+\gamma]\cap[0,1]$, so does their average, even if all scalar reports are adversarial; unlike median aggregation, this gives posterior-anchor containment rather than an honest-range guarantee over ordered intervention actions.

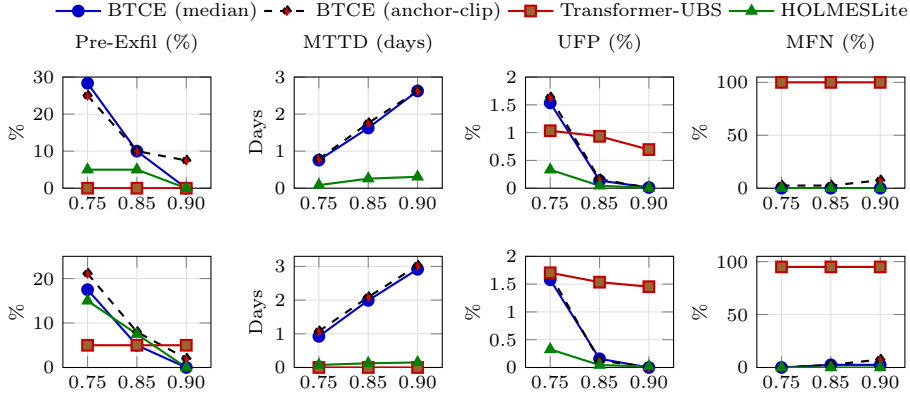
\begin{figure*}[t]
\centering
\begin{tikzpicture}
\begin{groupplot}[
  group style={group size=4 by 2, horizontal sep=1.1cm, vertical sep=0.9cm},
  width=0.29\textwidth,
  height=0.25\textwidth,
  grid=both,
  minor grid style={draw=black!7},
  major grid style={draw=black!12},
  tick label style={font=\scriptsize},
  label style={font=\scriptsize},
  title style={font=\scriptsize},
  legend style={font=\scriptsize, draw=none, fill=none},
  symbolic x coords={0.75, 0.85, 0.90},
  xtick=data,
  xticklabels={$0.75$,$0.85$,$0.90$},
  enlarge x limits=0.25, 
 ylabel shift={-3pt}, 
  every axis plot/.append style={thick},
  unbounded coords=jump,
]


\nextgroupplot[
  title={Pre-Exfil (\%)},
  ylabel={\%},
  ymin=0, ymax=30,
  legend to name=leg_userearly_new,
  legend columns=4
]
\addplot+[blue!75!black, mark=*] coordinates {(0.75,28.3333) (0.85,10.0000) (0.90,0.0000)};
\addplot+[black, dashed, mark=diamond*] coordinates {(0.75,25.0) (0.85,10.0) (0.90,7.5)};
\addplot+[red!75!black, mark=square*] coordinates {(0.75,0.0000) (0.85,0.0000) (0.90,0.0000)};
\addplot+[green!50!black, mark=triangle*] coordinates {(0.75,5.0000) (0.85,5.0000) (0.90,0.0000)};
\legend{BTCE (median), BTCE (anchor-clip), Transformer-UBS, HOLMESLite}

\nextgroupplot[
  title={MTTD (days)},
  ylabel={Days},
  ymin=0, ymax=3
]
\addplot+[blue!75!black, mark=*] coordinates {(0.75,0.7583) (0.85,1.6250) (0.90,2.6250)};
\addplot+[black, dashed, mark=diamond*] coordinates {(0.75,0.7878) (0.85,1.7578) (0.90,2.6279)};
\addplot+[red!75!black, mark=square*] coordinates {(0.75,nan) (0.85,nan) (0.90,nan)};
\addplot+[green!50!black, mark=triangle*] coordinates {(0.75,0.0833) (0.85,0.2583) (0.90,0.3083)};

\nextgroupplot[
  title={UFP (\%)},
  ylabel={\%},
  ymin=0, ymax=2.0
]
\addplot+[blue!75!black, mark=*] coordinates {(0.75,1.5351) (0.85,0.1378) (0.90,0.0125)};
\addplot+[black, dashed, mark=diamond*] coordinates {(0.75,1.6353) (0.85,0.1629) (0.90,0.0063)};
\addplot+[red!75!black, mark=square*] coordinates {(0.75,1.0338) (0.85,0.9336) (0.90,0.6955)};
\addplot+[green!50!black, mark=triangle*] coordinates {(0.75,0.3321) (0.85,0.0439) (0.90,0.0063)};

\nextgroupplot[
  title={MFN (\%)},
  ylabel={\%},
  ymin=0, ymax=105
]
\addplot+[blue!75!black, mark=*] coordinates {(0.75,0.0) (0.85,0.0) (0.90,0.0)};
\addplot+[black, dashed, mark=diamond*] coordinates {(0.75,2.5) (0.85,2.5) (0.90,7.5)};
\addplot+[red!75!black, mark=square*] coordinates {(0.75,100.0) (0.85,100.0) (0.90,100.0)};
\addplot+[green!50!black, mark=triangle*] coordinates {(0.75,0.0) (0.85,0.0) (0.90,0.0)};


\nextgroupplot[
  title={},
  xlabel={},
  ylabel={\%},
  ymin=0, ymax=25
]
\addplot+[blue!75!black, mark=*] coordinates {(0.75,17.5000) (0.85,5.000) (0.90,0.0000)};
\addplot+[black, dashed, mark=diamond*] coordinates {(0.75,21.1111) (0.85,8.0000) (0.90,2.0000)};
\addplot+[red!75!black, mark=square*] coordinates {(0.75,5.0000) (0.85,5.0000) (0.90,5.0000)};
\addplot+[green!50!black, mark=triangle*] coordinates {(0.75,15.0000) (0.85,7.5000) (0.90,0.0000)};

\nextgroupplot[
  title={},
  xlabel={},
  ylabel={Days},
  ymin=0, ymax=3.3
]
\addplot+[blue!75!black, mark=*] coordinates {(0.75,0.9250) (0.85,1.9917) (0.90,2.9167)};
\addplot+[black, dashed, mark=diamond*] coordinates {(0.75,1.0750) (0.85,2.0917) (0.90,3.0250)};
\addplot+[red!75!black, mark=square*] coordinates {(0.75,0.0000) (0.85,0.0000) (0.90,0.0000)};
\addplot+[green!50!black, mark=triangle*] coordinates {(0.75,0.0750) (0.85,0.1250) (0.90,0.1500)};

\nextgroupplot[
  title={},
  xlabel={},
  ylabel={\%},
  ymin=0, ymax=2.0
]
\addplot+[blue!75!black, mark=*] coordinates {(0.75,1.5789) (0.85,0.1566) (0.90,0.0000)};
\addplot+[black, dashed, mark=diamond*] coordinates {(0.75,1.6729) (0.85,0.1253) (0.90,0.0063)};
\addplot+[red!75!black, mark=square*] coordinates {(0.75,1.7043) (0.85,1.5351) (0.90,1.4536)};
\addplot+[green!50!black, mark=triangle*] coordinates {(0.75,0.3258) (0.85,0.0439) (0.90,0.0188)};

\nextgroupplot[
  title={},
  xlabel={},
  ylabel={\%},
  ymin=0, ymax=105
]
\addplot+[blue!75!black, mark=*] coordinates {(0.75,0.0) (0.85,2.5) (0.90,2.5)};
\addplot+[black, dashed, mark=diamond*] coordinates {(0.75,0.0) (0.85,2.5) (0.90,7.5)};
\addplot+[red!75!black, mark=square*] coordinates {(0.75,95.0) (0.85,95.0) (0.90,95.0)};
\addplot+[green!50!black, mark=triangle*] coordinates {(0.75,0.0) (0.85,0.0) (0.90,0.0)};

\end{groupplot}

\node[anchor=south] at ($(group c1r1.north)!0.5!(group c4r1.north) + (0,5mm)$)
{\pgfplotslegendfromname{leg_userearly_new}};

\end{tikzpicture}
\vspace{-2.5mm}
\caption{UserEarly baseline comparison across thresholds $\theta$, faceted by the Byzantine mix used to generate the closed-loop trajectories (top row $f=3$, bottom row $f=9$).}
\label{fig:userearly_baseline_plot}
\vspace{-3mm}
\end{figure*}
\section{Results and Discussion} 
\label{sec:empirical}
We evaluate on CERT r6.2~\cite{CERT2016} (January 2010--June 2011), using daily windows with $N=4000$ users and $0.125\%$ malicious prevalence ($5$ malicious users per window). Logs are mapped to continuous channel scores in $[0,1]$, strategic deviations follow Sec.~\ref{sec:btce-derivations}, and the committee has $m=25$ certifiers running Algorithm~\ref{alg:intervention}, of which $f$ are Byzantine ($|\mathcal C_{\mathrm{byz}}|=f$). We compare BTCE against Transformer-UBS trained under the same $5/4000$ imbalance and evaluated on the same closed-loop trajectories generated under each Byzantine mix \(f\). We also evaluate against HOLMESLite, a streaming provenance-style baseline using exponentially decayed z-score aggregation with cross-channel bonuses~\cite{milajerdi2019holmes}. For BTCE, $t_{\mathrm{intervene}}(u)=\inf\{t\ge t_0:\mu_u^{t,\mathrm{sys}}(M)\ge\mu_{\mathrm{thr}}\}$ (we write $\theta$ for the threshold value $\mu_{\mathrm{thr}}$ below); for baselines it is the first calibrated threshold crossing. With first exfiltration $t_{\mathrm{exfil}}(u)$ and $\Delta_u=t_{\mathrm{intervene}}(u)-t_{\mathrm{exfil}}(u)$, we report Pre-Exfil $(\Delta_u<0)$, $\Delta_u$ mean/median, MTTD, UFP, MFN, and F1 under \texttt{Period}, \texttt{UserEver}, and \texttt{UserEarly} units. Here MTTD (mean time to detection) is the average gap $t_{\mathrm{intervene}}-t_0$ from the detection-window start to intervention, UFP (user false-positive rate) is the share of benign users that ever receive an intervention, and MFN (malicious false-negative rate) is the share of malicious users never flagged. The rational ablation sets $\psi_{\mathrm{pb}}=1$; concept-drift runs use two counterfactual rollout replicates and introduce a benign covariate shift at $t_d\in\{5,10,15\}$ with shift magnitude $0.12$ and noise scale $0.03$. Shift magnitude is the benign mean displacement after drift; noise scale is the standard deviation of added per-channel perturbations.
Across these runs, the median-aggregation variant satisfies Proposition~\ref{prop:median-robust}: when $f<m/2$, the implemented action remains within the honest recommendation range; the anchor-clip variant satisfies posterior-anchor containment, with the aggregate score remaining within $[\widetilde\mu_u^t-\gamma,\widetilde\mu_u^t+\gamma]$.
\begin{table}[t]
\caption{Behavioral vs.\ rational BTCE on early-warning timing under median aggregation over 10 runs.
$\Delta=t_{\mathrm{intervene}}-t_{\mathrm{exfil}}$ and
MTTD $=t_{\mathrm{intervene}}-t_0$.}
\label{tab:beh_vs_rat}
\centering
\scriptsize
\setlength{\tabcolsep}{2.6pt}
\renewcommand{\arraystretch}{1.1}
\begin{tabular}{@{}llcccccccc@{}}
\toprule
 &  & \multicolumn{2}{c}{$\boldsymbol{\Delta}$ mean}
    & \multicolumn{2}{c}{$\boldsymbol{\Delta}$ median}
    & \multicolumn{2}{c}{\textbf{MTTD}}
    & \multicolumn{2}{c}{\textbf{Earlier by}} \\
\cmidrule(lr){3-4}\cmidrule(lr){5-6}\cmidrule(lr){7-8}\cmidrule(l){9-10}
$\boldsymbol{\theta}$ & \textbf{Model}
& $f=3$ & $f=9$
& $f=3$ & $f=9$
& $f=3$ & $f=9$
& $f=3$ & $f=9$ \\
\midrule
\multirow{2}{*}{0.75}
 & Behavioral & 0.55 & 0.40 & 0.90 & 0.90 & 0.76 & 0.93 & \multirow{2}{*}{1.70} & \multirow{2}{*}{1.76} \\
 & Rational   & 2.46 & 2.69 & 2.05 & 2.05 & 2.46 & 2.69 &       &       \\
\midrule
\multirow{2}{*}{0.85}
 & Behavioral & 1.73 & 1.76 & 2.00 & 1.95 & 1.63 & 1.99 & \multirow{2}{*}{3.39} & \multirow{2}{*}{2.62} \\
 & Rational   & 5.02 & 4.61 & 4.20 & 3.70 & 5.02 & 4.61 &       &       \\
\midrule
\multirow{2}{*}{0.90}
 & Behavioral & 2.88 & 2.88 & 3.10 & 2.90 & 2.63 & 2.92 & \multirow{2}{*}{4.50} & \multirow{2}{*}{4.01} \\
 & Rational   & 7.13 & 6.93 & 6.10 & 5.85 & 7.13 & 6.93 &       &       \\
\bottomrule
\end{tabular}
\end{table}
\paragraph{Findings.} We evaluate four key objectives below:

\begin{enumerate}[leftmargin=*,nosep]
\item Early detection with bounded user-level false alarms: At loose-to-moderate thresholds, where intervention is operationally meaningful, BTCE is the only method with useful pre-exfiltration coverage (Fig.~\ref{fig:userearly_baseline_plot}). At $\theta=0.75$, median BTCE achieves Pre-Exfil $28.3\%$ at $f=3$ and $17.5\%$ at $f=9$, with UFP $1.54\%$ and $1.58\%$, respectively. Transformer-UBS misses every malicious user in the $f=3$ runs and $95\%$ of malicious users in the $f=9$ runs; HOLMESLite's MTTD is small, but that speed buys little early warning because its alerts land within $0.08$ days of exfiltration on average, making it effectively post-hoc. At $\theta=0.90$, all methods lose meaningful pre-exfiltration coverage because the required evidence does not accumulate before the exfiltration window closes.  Fig.~\ref{fig:btce_f1_period_vs_userever} shows the within-BTCE tradeoff: looser thresholds improve Pre-Exfil and MTTD, conservative thresholds improve \texttt{UserEver}-F1, and intermediate thresholds maximize \texttt{Period}-F1. Anchor-clip improves \texttt{Period}-F1: in the closed-loop simulator, the aggregation rule changes the belief path, user actions, and attack-window labels, keeping the score above threshold across more of the kill-chain window, even when \texttt{UserEver}-F1 is similar. Median remains the headline mechanism because it trades this empirical aggressiveness for the honest-range guarantee of Proposition~\ref{prop:median-robust}.

\item Behavioral vs.\ rational: Removing present bias delays detection at every threshold and Byzantine mix (Table~\ref{tab:beh_vs_rat}). Behavioral BTCE has lower MTTD than rational BTCE throughout, confirming that impulsive escalation is the mechanism by which evidence drift becomes positive in finite time (Theorem~\ref{thm:early-detection}). At $f=9$, anchor-clip preserves stronger timing-sensitive classification in Fig.~\ref{fig:btce_f1_period_vs_userever}, but its absolute behavioral $\Delta_u$ and MTTD are comparable to or worse than median.
 
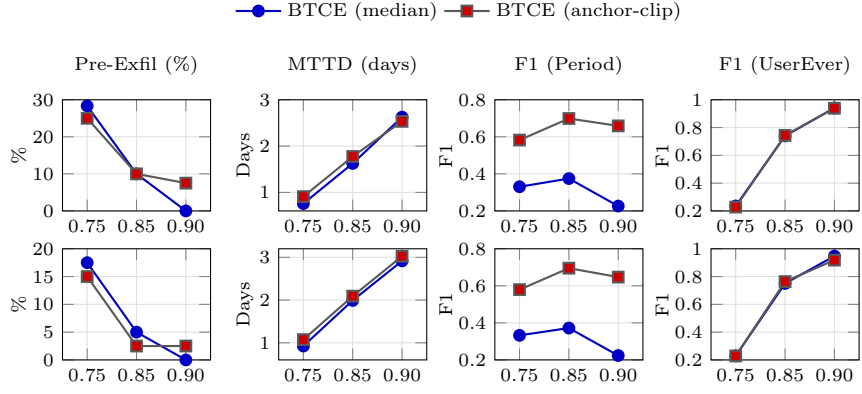
\begin{figure*}[t]
\centering
\begin{tikzpicture}
\begin{groupplot}[
  group style={
    group size=4 by 2, 
    horizontal sep=0.9cm, 
    vertical sep=0.5cm
  },
  width=0.29\textwidth,
  height=0.25\textwidth,
  grid=both,
  minor grid style={draw=black!7},
  major grid style={draw=black!12},
  symbolic x coords={0.75, 0.85, 0.90},
  xtick=data,
  xticklabels={$0.75$,$0.85$,$0.90$},
  enlarge x limits=0.25,
  tick label style={font=\scriptsize},
  label style={font=\scriptsize},
  title style={font=\scriptsize},
  legend style={font=\scriptsize, draw=none, fill=none},
  every axis plot/.append style={thick},
]
\nextgroupplot[
  title={Pre-Exfil (\%)},
  ylabel={\%},
  ymin=0, ymax=30,
  ylabel shift={-4pt},
  legend to name=leg_f1_line,
  legend columns=2
]
\addplot+[blue!75!black, mark=*] coordinates {
  (0.75,28.3333) (0.85,10.0) (0.90,0.0000)
};
\addplot+[gray!70!black, mark=square*] coordinates {
  (0.75,25.0) (0.85,10.0) (0.90,7.5)
};
\legend{BTCE (median), BTCE (anchor-clip)}
\nextgroupplot[
  title={MTTD (days)},
  ylabel={Days},
  ymin=0.6, ymax=3,
  ylabel shift={-4pt}
]
\addplot+[blue!75!black, mark=*] coordinates {
  (0.75,0.7583) (0.85,1.6250) (0.90,2.6250)
};
\addplot+[gray!70!black, mark=square*] coordinates {
  (0.75,0.9083) (0.85,1.7750) (0.90,2.5333)
};
\nextgroupplot[
  title={F1 (Period)},
  ylabel={F1},
  ymin=0.2, ymax=0.8,
  ylabel shift={-4pt}
]
\addplot+[blue!75!black, mark=*] coordinates {
  (0.75,0.3304) (0.85,0.3746) (0.90,0.2261)
};
\addplot+[gray!70!black, mark=square*] coordinates {
  (0.75,0.5825) (0.85,0.6988) (0.90,0.6591)
};
\nextgroupplot[
  title={F1 (UserEver)},
  ylabel={F1},
  ymin=0.2, ymax=1.0,
  ylabel shift={-4pt}
]
\addplot+[blue!75!black, mark=*] coordinates {
  (0.75,0.2356) (0.85,0.7398) (0.90,0.9409)
};
\addplot+[gray!70!black, mark=square*] coordinates {
  (0.75,0.2277) (0.85,0.7435) (0.90,0.9387)
};
\nextgroupplot[
  title={},
  xlabel={},
  ylabel={\%},
  ymin=0, ymax=20,
  ylabel shift={-4pt}
]
\addplot+[blue!75!black, mark=*] coordinates {
  (0.75,17.5) (0.85,5.0000) (0.90,0.0000)
};
\addplot+[gray!70!black, mark=square*] coordinates {
  (0.75,15.0) (0.85,2.5) (0.90,2.5)
};
\nextgroupplot[
  title={},
  xlabel={},
  ylabel={Days},
  ymin=0.6, ymax=3.2,
  ylabel shift={-4pt}
]
\addplot+[blue!75!black, mark=*] coordinates {
  (0.75,0.9250) (0.85,1.9917) (0.90,2.9167)
};
\addplot+[gray!70!black, mark=square*] coordinates {
  (0.75,1.0750) (0.85,2.0917) (0.90,3.0250)
};
\nextgroupplot[
  title={},
  xlabel={},
  ylabel={F1},
  ymin=0.2, ymax=0.8,
  ylabel shift={-4pt}
]
\addplot+[blue!75!black, mark=*] coordinates {
  (0.75,0.3325) (0.85,0.3720) (0.90,0.2235)
};
\addplot+[gray!70!black, mark=square*] coordinates {
  (0.75,0.5793) (0.85,0.6952) (0.90,0.6470)
};
\nextgroupplot[
  title={},
  xlabel={},
  ylabel={F1},
  ymin=0.2, ymax=1.0,
  ylabel shift={-4pt}
]
\addplot+[blue!75!black, mark=*] coordinates {
  (0.75,0.2289) (0.85,0.7505) (0.90,0.9474)
};
\addplot+[gray!70!black, mark=square*] coordinates {
  (0.75,0.2290) (0.85,0.7640) (0.90,0.9183)
};
\end{groupplot}
\node[anchor=south] at ($(group c1r1.north)!0.5!(group c4r1.north) + (0,8mm)$)
{\pgfplotslegendfromname{leg_f1_line}};
\end{tikzpicture}
\caption{BTCE tradeoff between timing-sensitive detection and user-level capture (top row $f=3$, bottom row $f=9$). Lower thresholds ($\theta$) improve Pre-Exfil and MTTD; more conservative thresholds increase \texttt{UserEver}-F1.}
\label{fig:btce_f1_period_vs_userever}
\vspace{-3mm}
\end{figure*}
\item Equilibrium compliance via one-period regret: We audit on-path states using $\mathrm{Regret}(I_u^t)=\max_{x\in A_u}\{W(x\mid I_u^t)-W(\hat{x}\mid I_u^t)\}$, estimating $W$ by Monte Carlo rollouts of an obedient baseline against a one-period deviation. Fig.~\ref{fig:merger_regret_ablation} top panel reports the audit as excess regret normalized by peak no-drift regret for drift-perturbed runs. In no-drift runs, non-target users exhibit stable regret, while target regret rises during the exfiltration window (Day~7--20) and declines after intervention, supporting restoration of obedience on the active population (Lemma~\ref{lem:oneshot} and Theorem~\ref{thm:obedience}). 
 
\item Robustness to concept drift: We introduce a single change-point covariate shift in the benign signal distribution at $t_d\in\{5,10,15\}$ while holding the attack process fixed, a stylized merger or system migration. MTTD remains bounded across drift start times while UFP grows with longer exposure to the shifted benign regime (Fig.~\ref{fig:merger_regret_ablation} bottom panel). Covariate shift perturbs non-target excess regret by at most $1.3\%$ of peak no-drift regret throughout the horizon, while target excess regret dips to roughly $-41\%$ mid-window and contracts to under $-10\%$ by day~29, indicating that the median mechanism re-enters a reduced post-intervention equilibrium even under benign-distribution shift, but at the cost of a larger false-positive budget when drift arrives earlier.
\end{enumerate}

\section{Conclusion}
BTCE achieves up to $28.3\%$ pre-exfiltration detection at $1.54\%$ UFP and detects $1.7$--$4.5$ days earlier than a rational ablation, while both a transformer baseline and a streaming provenance system reach near-zero pre-exfiltration coverage under identical constraints. Each outcome traces back to a formal guarantee: calibrated present-bias sanctions sustain current-self obedience, producing the impulsive escalation that drives positive evidence drift and bounds expected detection time before exfiltration; one-shot obedience rules out patient multi-period alternatives; and median aggregation preserves $17.5\%$ pre-exfiltration coverage even under nine Byzantine certifiers. Under covariate shift the mechanism re-enters a reduced post-intervention equilibrium, with MTTD bounded across all drift start times, though earlier drift widens the false-positive budget. Together these results show that grounding insider-threat detection in equilibrium coordination makes early warning auditable and creates opportunities for robust implementation in a way that purely predictive baselines cannot match.
\begin{figure*}[t]
\centering
\begin{tikzpicture}
\begin{groupplot}[
  group style={
    group size=2 by 2,
    horizontal sep=1.55cm,
    vertical sep=1.5cm 
  },
  width=0.48\textwidth,
  height=0.26\textwidth,
  grid=both,
  minor grid style={draw=black!7},
  major grid style={draw=black!12},
  tick label style={font=\scriptsize},
  label style={font=\scriptsize},
  title style={font=\scriptsize},
  legend style={font=\scriptsize, draw=none, fill=none},
  ylabel shift={-2pt},
]

\nextgroupplot[
  title={Restoration (Targets)},
  ylabel={Excess regret (\%)},
  xlabel={Day},
  xtick={0,5,10,15,20,25,29},
  ymin=-45, ymax=5,
]
\addplot[black!35, dashed, thin, forget plot] coordinates {(5,-45) (5,5)};
\addplot[black!35, dashed, thin, forget plot] coordinates {(10,-45) (10,5)};
\addplot[black!35, dashed, thin, forget plot] coordinates {(15,-45) (15,5)};

\addplot+[red, thick, mark=none] coordinates {
 (0,0) (1,0) (2,0) (3,0) (4,0) (5,0) (6,0) (7,0) (8,0) (9,0)
 (10,0) (11,0) (12,0) (13,0) (14,0) (15,0) (16,0) (17,0) (18,0) (19,0)
 (20,0) (21,0) (22,0) (23,0) (24,0) (25,0) (26,0) (27,0) (28,0) (29,0)
};

\addplot+[black, thick, dashed, mark=none] coordinates {
 (0,0.0000) (1,0.0000) (2,0.0000) (3,0.0000) (4,0.0000)
 (5,0.0000) (6,-0.3779) (7,-0.6540) (8,-0.3613) (9,-0.3329)
 (10,-1.9019) (11,-2.3114) (12,-3.7350) (13,-6.0610) (14,-8.6578)
 (15,-16.7858) (16,-24.5419) (17,-28.9071) (18,-29.4752) (19,-30.8076)
 (20,-38.3966) (21,-38.5974) (22,-39.6464) (23,-41.2372) (24,-38.2433)
 (25,-23.8397) (26,-19.1468) (27,-10.8166) (28,-8.5230) (29,-8.5217)
};

\addplot+[black, thick, dotted, mark=none] coordinates {
 (0,0.0000) (1,0.0000) (2,0.0000) (3,0.0000) (4,0.0000)
 (5,0.0000) (6,0.0000) (7,0.0000) (8,0.0000) (9,0.0000)
 (10,0.0000) (11,-0.5219) (12,-0.1618) (13,-2.0683) (14,-2.1661)
 (15,-6.0260) (16,-10.9699) (17,-16.0823) (18,-17.3982) (19,-16.3075)
 (20,-22.8356) (21,-24.6231) (22,-27.3606) (23,-33.9019) (24,-27.8693)
 (25,-20.2553) (26,-21.8830) (27,-10.9567) (28,-9.7537) (29,-8.5217)
};

\addplot+[black, thick, dashdotted, mark=none] coordinates {
 (0,0.0000) (1,0.0000) (2,0.0000) (3,0.0000) (4,0.0000)
 (5,0.0000) (6,0.0000) (7,0.0000) (8,0.0000) (9,0.0000)
 (10,0.0000) (11,0.0000) (12,0.0000) (13,0.0000) (14,0.0000)
 (15,0.0000) (16,-0.0456) (17,0.0390) (18,-3.1590) (19,-4.3157)
 (20,-6.4527) (21,-17.6399) (22,-19.8794) (23,-25.2243) (24,-27.5687)
 (25,-15.9653) (26,-15.6357) (27,-4.7299) (28,-4.8759) (29,-4.8699)
};

\nextgroupplot[
  title={Restoration (Non-targets)},
  ylabel={Excess regret (\%)},
  xlabel={Day},
  xtick={0,5,10,15,20,25,29},
  ymin=-1.5, ymax=0.1,
  yticklabel style={/pgf/number format/fixed}
]
\addplot[black!35, dashed, thin, forget plot] coordinates {(5,-1.5) (5,0.1)};
\addplot[black!35, dashed, thin, forget plot] coordinates {(10,-1.5) (10,0.1)};
\addplot[black!35, dashed, thin, forget plot] coordinates {(15,-1.5) (15,0.1)};

\addplot+[blue, thick, mark=none] coordinates {
 (0,0) (1,0) (2,0) (3,0) (4,0) (5,0) (6,0) (7,0) (8,0) (9,0)
 (10,0) (11,0) (12,0) (13,0) (14,0) (15,0) (16,0) (17,0) (18,0) (19,0)
 (20,0) (21,0) (22,0) (23,0) (24,0) (25,0) (26,0) (27,0) (28,0) (29,0)
};

\addplot+[black!60, thick, dashed, mark=none] coordinates {
 (0,0.0000) (1,0.0000) (2,0.0000) (3,0.0000) (4,0.0000)
 (5,0.0000) (6,-0.6266) (7,-0.9602) (8,-1.1088) (9,-1.1601)
 (10,-1.1750) (11,-1.1925) (12,-1.2722) (13,-1.2236) (14,-1.1979)
 (15,-1.1750) (16,-1.1439) (17,-1.1979) (18,-1.2438) (19,-1.1831)
 (20,-1.0588) (21,-1.2425) (22,-1.2276) (23,-1.3033) (24,-1.0696)
 (25,-1.0858) (26,-1.0534) (27,-1.1236) (28,-1.1344) (29,-1.1966)
};

\addplot+[black!60, thick, dotted, mark=none] coordinates {
 (0,0.0000) (1,0.0000) (2,0.0000) (3,0.0000) (4,0.0000)
 (5,0.0000) (6,0.0000) (7,0.0000) (8,0.0000) (9,0.0000)
 (10,0.0000) (11,-0.5159) (12,-0.8603) (13,-0.9724) (14,-1.0480)
 (15,-1.1290) (16,-1.1074) (17,-1.2141) (18,-1.1628) (19,-1.1925)
 (20,-1.1804) (21,-1.1371) (22,-1.1736) (23,-1.2020) (24,-1.0791)
 (25,-1.1412) (26,-1.2492) (27,-1.2371) (28,-1.1398) (29,-1.1993)
};

\addplot+[black!60, thick, dashdotted, mark=none] coordinates {
 (0,0.0000) (1,0.0000) (2,0.0000) (3,0.0000) (4,0.0000)
 (5,0.0000) (6,0.0000) (7,0.0000) (8,0.0000) (9,0.0000)
 (10,0.0000) (11,0.0000) (12,0.0000) (13,0.0000) (14,0.0000)
 (15,0.0000) (16,-0.4835) (17,-0.7441) (18,-0.9616) (19,-1.0399)
 (20,-1.0683) (21,-1.1763) (22,-1.1682) (23,-1.2114) (24,-1.1223)
 (25,-1.2155) (26,-1.2857) (27,-1.2114) (28,-1.1601) (29,-1.1034)
};

\nextgroupplot[
  title={MTTD},
  xlabel={Drift start},
  ylabel={Days},
  symbolic x coords={none,5,10,15},
  xtick=data,
  ymin=0, ymax=5,
]
\addplot+[
  red,
  thick,
  mark=*,
  error bars/.cd,
    y dir=both, y explicit,
] coordinates {
  (none,2.8844) +- (0,0.3509)
  (5,2.1510) +- (0,0.3504)
  (10,2.0980) +- (0,0.5492)
  (15,2.0969) +- (0,0.3000)
};

\nextgroupplot[
  title={User false positives},
  xlabel={Drift start},
  ylabel={\%},
  symbolic x coords={none,5,10,15},
  xtick=data,
  ymin=0, ymax=45,
]
\addplot+[
  black!65,
  thick,
  mark=*,
  error bars/.cd,
    y dir=both, y explicit,
] coordinates {
  (none,0.0251) +- (0,0.0167)
  (5,38.3609) +- (0,1.0322)
  (10,30.4912) +- (0,0.9358)
  (15,22.0927) +- (0,0.7593)
};

\end{groupplot}
\end{tikzpicture}
\vspace{-2mm}
\caption{Concept-drift stress test for median aggregation. Top panels show restoration as excess one-step regret normalized by peak no-drift regret; solid lines are the no-drift reference, drifts start at days 5, 10, and 15. Bottom panels report MTTD and user false positives as a function of drift start time.}
\label{fig:merger_regret_ablation}
\vspace{-4mm}
\end{figure*}
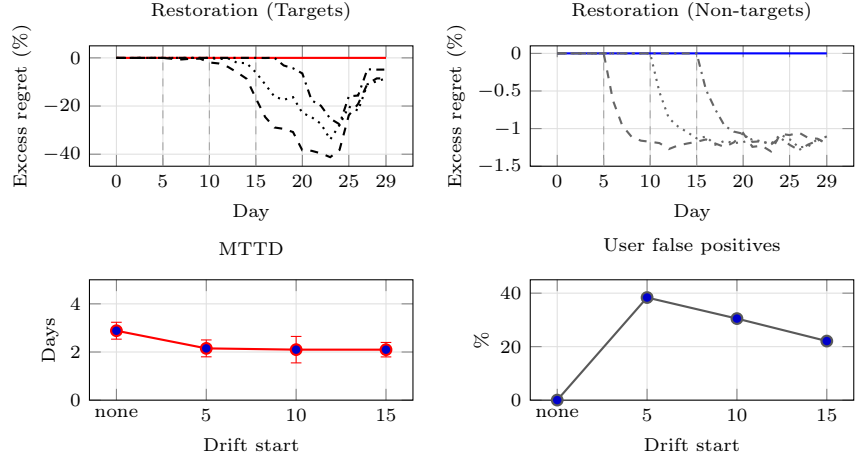

\paragraph{Limitations and future work}
The early-detection guarantee assumes a positive drift floor $\lambda_{\min}>0$: a patient adversary can stay within the benign noise band and drive $\lambda_{\min}$ toward zero, at the cost of exfiltration throughput. Documented insiders do not appear to behave this way, present bias produces impulsive escalation rather than slow extraction, but a formal stealth--throughput tradeoff parameterized by $G_{\mathrm{sig}}$ remains open. Because drift events are typically scheduled (mergers and acquisitions), a merger-aware refit of the benign emission model on post-$t_d$ samples is a natural direction for restoring the UFP budget. ITS weights and conditional probability tables should be periodically recalibrated; the DAG $G_{\mathrm{sig}}$ itself is fixed by domain knowledge (\S\ref{def:exploit-chain}). We assume trusted implementation of $g_\kappa$ and the belief filter.

\appendix
\section{Type Evolution Dynamics}
\label{app:type-evolution}

\begin{definition}[Reference-point dynamics]
\label{def:ref-point-dynamics}
Set $r_u^0=0$ (neutral baseline). Let $\lambda_r\in[0,1]$, $\kappa_r\ge 0$, and
$\mathbb{I}_{\mathrm{adverse}}:A_c\to\{0,1\}$. The pre-discretization update is $\widetilde r_u^{t+1}
\;=\;
(1-\lambda_r)\,r_u^t
\;+\;
\lambda_r\,\kappa_r\,\mathbb{I}_{\mathrm{adverse}}(d_u^t)$. Since $B$ is finite, $r_u^{t+1}$ is obtained by stochastic rounding $\widetilde r_u^{t+1}$ to the two adjacent
points of $\mathcal{R}=\{r^{(1)}<\cdots<r^{(K_r)}\}\subseteq\mathbb{R}$ with probabilities proportional to
distance (unbiased). Absent adverse decisions $r_u^t$ decays toward $0$; persistent adverse outcomes raise it toward $\kappa_r$, deepening the loss domain and triggering risk-seeking convexity in $v(\cdot)$.
\end{definition}
\begin{definition}[Present-bias dynamics]
\label{def:present-bias-dynamics}
Let $\psi_{\mathrm{impulsive}}<\psi_{\mathrm{stable}}$ in $(0,1]$, threshold $\gamma_{\mathrm{stress}}\in\mathbb{R}$,
and stress proxy $\omega_u^t\coloneqq r_u^t-\rho(x_u^t)$ (since $x_u^t$ is unobserved, $\omega_u^t$ is inferred
from the public signal). The pre-discretization update is
$\widetilde\psi_{\mathrm{pb},u}^{t+1}
=
\begin{cases}
\psi_{\mathrm{impulsive}} & \text{if } \omega_u^t \ge \gamma_{\mathrm{stress}},\\
\psi_{\mathrm{stable}}   & \text{otherwise,}
\end{cases}$
and $\psi_{\mathrm{pb},u}^{t+1}$ is obtained by the same stochastic-rounding rule applied to
$\Psi=\{\psi^{(1)}<\cdots<\psi^{(K_\psi)}\}\subseteq(0,1]$.
\end{definition}

\noindent Definitions~\ref{def:ref-point-dynamics}--\ref{def:present-bias-dynamics} instantiate
$\tau_B(b_u^{t+1}\mid b_u^t,\theta_u^t,x_u^t,d_u^t)\in\Delta(B)$ from Def.~\ref{def:markov-behavioral}.
A deterministic specialization sets
$\tau_B(b'\mid b,\theta,x,d)=\mathbb{I}\!\{b'=F(b,\theta,x,d)\}$,
where $F:B\times\Theta_{\mathrm{user}}\times A_u\times A_c\to B$ applies the above updates and rounds to the grid.

\bibliographystyle{splncs04}
\bibliography{temporal}

\paragraph{Disclosure of Interests.} The authors have no competing interests to declare that are relevant to the content of this article.
\end{document}